\documentclass[aip, jmp, numerical, preprint, onecolumn, 12pt, groupedaddress]{revtex4-1}
\usepackage{graphicx}
\usepackage{amsmath}
\usepackage{amsthm}
\usepackage{amsfonts}
\usepackage{amssymb}
\usepackage[caption=false]{subfig}
\usepackage{bm}
\usepackage{natbib}

\newcommand{\ud}{\,\textrm{d}}

\newcommand{\RR}{\mathbb{R}}
\newcommand{\CCC}{\mathbb{C}}

\newcommand{\bxi}{\boldsymbol{\xi}}
\newcommand{\bx}{\boldsymbol{x}}
\newcommand{\by}{\boldsymbol{y}}
\newcommand{\balpha}{\boldsymbol{\alpha}}
\newcommand{\bz}{\boldsymbol{z}}

\newcommand{\bu}{\boldsymbol{u}}

\newcommand{\bzero}{\boldsymbol{0}}

\newcommand{\bnabla}{\boldsymbol{\nabla}}

\newcommand{\bmm}{\boldsymbol{m}}

\newcommand{\ii}{\textrm{i}}

\newcommand{\spn}{\text{Span}}

\newcommand{\co}{\text{co}}

\newcommand{\bos}{\boldsymbol{s}}
\newcommand{\bphi}{\boldsymbol{\phi}}
\newcommand{\bor}{\boldsymbol{r}}

\newcommand{\bY}{\boldsymbol{Y}}
\newcommand{\bM}{\boldsymbol{M}}
\newcommand{\bI}{\boldsymbol{I}}
\newcommand{\bH}{\boldsymbol{H}}

\newcommand{\rme}{\textrm{e}}
\newcommand{\rmi}{\text{i}}
\newtheorem{theorem}{Theorem}
\newtheorem{lemma}[theorem]{Lemma}
\newtheorem{corollary}[theorem]{Corollary}
\theoremstyle{definition} \newtheorem*{remark}{Remark}

\begin{document}


\title{Two-step asymptotics of scaled Dunkl processes}

\date{\today}

\author{Sergio Andraus}
\email{andraus@spin.phys.s.u-tokyo.ac.jp}

\author{Seiji Miyashita}
\affiliation{Graduate School of Science, The University of Tokyo}


\begin{abstract}
Dunkl processes are generalizations of Brownian motion obtained by using the differential-difference operators known as Dunkl operators as a replacement of spatial partial derivatives in the heat equation.
Special cases of these processes include Dyson's Brownian motion model and the Wishart-Laguerre eigenvalue processes, which are well-known in random matrix theory.
It is known that the dynamics of Dunkl processes is obtained by transforming the heat kernel using Dunkl's intertwining operator.
It is also known that, under an appropriate scaling, their distribution function converges to a steady-state distribution which depends only on the coupling parameter $\beta$ as the process time $t$ tends to infinity. 
We study scaled Dunkl processes starting from an arbitrary initial distribution, and we derive expressions for the intertwining operator in order to calculate the asymptotics of the distribution function in two limiting situations. In the first one, $\beta$ is fixed and $t$ tends to infinity (approach to the steady state), and in the second one, $t$ is fixed and $\beta$ tends to infinity (strong-coupling limit).
We obtain the deviations from the limiting distributions in both of the above situations, and we find that they are caused by the two different mechanisms which drive the process, namely, the drift and exchange mechanisms. 
We find that the deviation due to the drift mechanism decays as $t^{-1}$, while the deviation due to the exchange mechanism decays as $t^{-1/2}$.
\end{abstract}
\maketitle

 

\section{Introduction}\label{intro}


The simple diffusion process is one of the most fundamental processes in physics, and it is modeled by Brownian motion.\cite{karatzasshreve91} The transition probability density (TPD) of Brownian motion, known as the heat kernel, obeys the heat equation.
Dunkl processes \cite{roslervoit98} are generalizations of multi-dimensional Brownian motion achieved through the use of Dunkl operators.\cite{dunkl89, dunklxu}
Dunkl operators consist of a differential operation with respect to a coordinate and of a sum of difference operations with respect to reflections defined by a finite set of vectors known as ``root system,'' as will be explained in the next section (Sec.~\ref{definitions}). This root system introduces the so-called Weyl chambers, which are disjoint portions of Euclidean space which are related to each other by the above reflections.
Dunkl processes are defined by the time evolution of their TPDs, which is given by a heat equation in which the Laplacian operator is replaced by the sum of the squares of Dunkl operators (the Dunkl heat equation). 
Because Dunkl operators contain differential and difference terms, the Dunkl heat equation contains a diffusion term, a drift term which drives the process away from the walls of the Weyl chambers, and a difference term among the Weyl chambers. The diffusion and drift terms drive the process within each of the Weyl chambers separately, while the difference term makes the process jump from one Weyl chamber to another, causing the process to relax toward a symmetry called ``$W$-invariance.'' We call the former ``drift'' mechanism, and the latter ``exchange'' mechanism. See Sec.~\ref{definitions} for details.

The relationship between the usual Brownian motion and Dunkl processes is formalized by the intertwining operator $V$, introduced by Dunkl in Ref.~\onlinecite{dunkl91}. The intertwining operator is a functional which is uniquely defined by the way it relates differential operators and Dunkl operators. In fact, $V$ transforms the heat equation into the Dunkl heat equation. Therefore, the solution of the Dunkl process, its TPD, is given by the action of $V$ on the solution of Brownian motion. We may even say that the dynamics of Dunkl processes are encoded in $V$. However, the explicit form of $V$ is unknown in general,\cite{rosler08, maslouhiyoussfi09} and although significant progress has been achieved recently,\cite{deleavaldemniyoussfi15} the study of Dunkl processes requires explicit derivations of the action of the intertwining operator for particular cases.

One of the most important properties of Dunkl processes is that, depending on the type of Dunkl operators under consideration, their continuous or ``radial'' component,\cite{gallardoyor05} which is the continuous motion of the process within the Weyl chambers, can be specialized to several well-known families of stochastic processes. In general, the norm, i.e., the distance from the origin of a Dunkl process, is given by a Bessel process.\cite{gallardoyor08} In addition, Dunkl operators of type $A_{N-1}$ produce a family of radial Dunkl processes which is mathematically equivalent to Dyson's Brownian motion model \cite{dyson62,demni08A} (henceforth referred to as Dyson's model). Dyson's model has been studied in relation with Fisher's vicious walker model,\cite{fisher84, katoritanemura02, katoritanemura07} polymer networks,\cite{degennes68,essamguttmann95} level statistics of atomic nuclei,\cite{bohigashaqpandey85} the Kardar-Parisi-Zhang universality class,\cite{prahoferspohn00,imamurasasamoto05,takeuchisano10,schehr12} traffic statistics,\cite{baikborodindeiftsuidan06} combinatorics and representation theory \cite{guttmannowczarekviennot98, krattenthalerguttmannviennot00, fulton} among many others.
Similarly, Dunkl operators of type $B_N$ give a family of radial Dunkl processes which corresponds to the eigenvalues of the Wishart and Laguerre processes.\cite{bru91,konigoconnell01, katoritanemura11} These multivariate stochastic processes are related to the QCD Dirac operator.\cite{verbaarschotzahed93} They have been studied as the eigenvalue processes of matrix-valued Brownian motions with chirality,\cite{katoritanemura04} and they are one example of the application of a multidimensional generalization of the Yamada-Watanabe theorem.\cite{graczykmalecki13} Dunkl operators themselves have also been used outside of stochastic processes, e.g., in the study of the Calogero-Moser-Sutherland systems,\cite{bakerforrester97, bakerdunklforrester, khastgir00} in a generalization of the quantum harmonic oscillator in multiple dimensions \cite{genestvinetzhedanov14} and also in supersymmetric quantum mechanics with reflections.\cite{postvinetzhedanov11}

It is noted that Dyson's model and the Wishart-Laguerre processes are matrix-valued processes indexed by the parameter $\beta$, which depends on the type of symmetry imposed on the entries of their corresponding matrices.\cite{mehta04, forrester10} When these matrices are real symmetric, complex Hermitian or quaternion self-dual, the parameter $\beta$ takes the values $1,2$ or $4$, respectively. In addition, it is known that the eigenvalues of these processes behave as particles in one dimensional space which repel mutually through a logarithmic potential, and $\beta$ is regarded as a coupling constant of interaction between the particles. Although the radial Dunkl processes of type $A_{N-1}$ and $B_N$ are well-defined for all $\beta>0$ and they share the stochastic differential equation of Dyson's model and the Wishart-Laguerre processes,\cite{demni08A} they do not have a known matrix-valued representation in the cases where $\beta$ is not equal to $1,2$ or 4. In our previous work,\cite{andrauskatorimiyashita12,andrauskatorimiyashita14} we examined Dyson's model and the Wishart-Laguerre processes through their formulation as radial Dunkl processes. 

In this paper, we study the distribution of an arbitrary Dunkl process whose space variables $\by$ have been scaled as $\by=\sqrt{\beta t}\bY$, where $t$ is the time-duration of the process. We calculate the asymptotics of the scaled distribution in two scenarios. 
Our first result (Theorem~\ref{steadystatetheorem}) states that, when $\beta > 0$ is fixed and $t$ tends to infinity, the distribution of of the process approaches a steady-state distribution with a first-order correction which decays with time as $t^{-1/2}$. This correction is a direct consequence of Lemma~\ref{vbetaonlinearfns}, which gives the action of $V$ on linear functions. Because the steady-state distribution is $W$-invariant, and the correction depends directly on the asymmetry (i.e., non-$W$-invariance) of the initial distribution, our result implies that this part of the relaxation process is due to the exchange mechanism.

Our second result (Theorem~\ref{TheoremFreezingLimit}) concerns the strong-coupling asymptotics of the scaled process, where $t>0$ is fixed and $\beta$ tends to infinity. In this case, the process distribution can be approximated by a sum of Gaussians centered at a set of points known as the ``peak set'' \cite{dunkl89B} of the particular type of Dunkl process considered. Finite-$\beta$ corrections to the center and the width of the approximating Gaussians are found to decay as $(\beta t)^{-1}$. In addition, the coefficients of the Gaussians are found to be different, but they converge to equal values as $(\beta t)^{-1/2}$. This result is obtained from Lemma~\ref{FreezingLimitDunklKernelFullRank}, which gives the action of $V$ on the exponential function when $\beta$ tends to infinity. From our results, we distinguish the two relaxation mechanisms in concrete terms. The relaxation due to the drift mechanism is found to be responsible for the width and position of each of the approximating Gaussians, while the exchange mechanism is found to be responsible for the relaxation of the height of the Gaussians.

This paper is organized as follows: in Sec.~\ref{definitions} we review the definitions of Dunkl operators, Dunkl processes and all related quantities. In Sec.~\ref{mainresults}, we give our results for the approach to the steady state (Theorem~\ref{steadystatetheorem} and Lemma~\ref{vbetaonlinearfns}) and the strong-coupling asymptotics (Theorem~\ref{TheoremFreezingLimit} and Lemma~\ref{FreezingLimitDunklKernelFullRank}). We illustrate these results for the case of the one-dimensional Dunkl process, for which the TPD is known explicitly. In Sec.~\ref{proofs}, we give the proof of our results. Finally, we discuss these results and propose a few related open problems in Sec.~\ref{conclusions}.

\section{Dunkl operators, Dunkl processes and the Intertwining operator}\label{definitions}
We briefly review the definition of Dunkl operators and other necessary mathematical objects. For more details, see Refs.~\onlinecite{dunklxu,rosler08}. 

Let us denote the reflection of the vector $\bx\in\RR^N$ along the vector $\balpha\in\RR^N$ by
\begin{equation}
\sigma_{\balpha}\bx:=\bx-2\frac{\balpha\cdot\bx}{\balpha\cdot\balpha}\balpha.
\end{equation}
A root system is a finite set of vectors, called roots, which is defined by the property that it remains unchanged if all of its elements are reflected along any particular root. In mathematical terms, a set of vectors $R$ is a root system if the set $\sigma_{\balpha} R:=\{\sigma_{\balpha}\bxi : \bxi\in R\}$ has the property that
\begin{equation}
\sigma_{\balpha} R=R,\ \forall\balpha\in R.
\end{equation}
In this paper, we impose the condition that the equation $a\bxi=\balpha$, for $\balpha,\bxi\in R$, implies that $a=\pm1$. Root systems that satisfy this condition are called ``reduced''. We define the positive subsystem $R_+=\{\balpha\in R : \balpha\cdot\bmm>0\}$ by choosing an arbitrary vector $\bmm$ such that $\bmm\cdot\balpha\neq 0$ for any root $\balpha$. Although the positive subsystem is chosen arbitrarily, the definitions that follow do not depend on the choice of $\bmm$.

For every root system, there is a group which is formed by all the reflections $\{\sigma_{\balpha}\}_{\balpha\in R}$ and their compositions. We denote this group by $W$. A Weyl chamber is defined as a connected subset of $\RR^N$ whose elements $\bx$ satisfy $\balpha\cdot\bx\neq 0$ for every root $\balpha$. Let us denote the number of elements in $W$ by $|W|$. Because each Weyl chamber is related to the others through the action of the elements of $W$, it follows that there are $|W|$ Weyl chambers. A parameter called ``multiplicity'' is assigned to each disjoint orbit of the roots $\balpha$ under the action of the elements of $W$, and the set of multiplicities is summarized as a function $k:R\to\CCC$ with the property that 
\begin{equation}\label{multiplicityfunctioncondition}
k(\sigma_{\balpha}\bxi)=k(\bxi)
\end{equation}
for $\balpha,\bxi\in R$. The multiplicities are parameters that are chosen arbitrarily, and in the present paper we assume that they are all real and positive, $k(\balpha)>0,\ \forall \balpha\in R$.

Let us denote by $\alpha_i$ the $i$th component of $\balpha$, and let us consider a differentiable function $f(\bx)$. Then, Dunkl operators are defined by
\begin{equation}
T_i f(\bx)=\frac{\partial}{\partial x_i}f(\bx)+\sum_{\balpha\in R_+}\alpha_ik(\balpha)\frac{[1-\sigma_{\balpha}]f(\bx)}{\balpha\cdot\bx},\ i=1,\ldots,N,
\end{equation}
where $\sigma_{\balpha}f(\bx)=f(\sigma_{\balpha}\bx)$, and for $\rho\in W,$ $\rho f(\bx)=f(\rho^{-1}\bx)$. If $f(\bx)$ has continuous second derivatives, then $T_iT_jf(\bx)=T_jT_if(\bx)$. In addition, the ``Dunkl Laplacian''\cite{dunklxu} is given by
\begin{equation}\label{dunkllaplacian}
\sum_{i=1}^NT_i^2 f(\bx)=\Delta f(\bx)+2\sum_{\balpha\in R_+}k(\balpha)\Big[\frac{\balpha\cdot\bnabla}{\balpha\cdot\bx}f(\bx)-\frac{\alpha^2}{2}\frac{1-\sigma_{\balpha}}{(\balpha\cdot\bx)^2}f(\bx)\Big],
\end{equation}
where $\Delta=\sum_{i=1}^N(\partial/\partial x_i)^2$ denotes the Laplacian operator, $\bnabla=(\partial/\partial x_1,\ldots,\partial/\partial x_N)^T$ denotes the gradient operator and $x=\sqrt{|\bx|^2}=\sqrt{\bx\cdot\bx}$ whenever no confusion arises.

Consider a stochastic process given by the TPD $p(t,\by|\bx)$, which represents the probability density that a process that starts at $\bx=(x_1,\ldots,x_N)^T$ reaches the position $\by=(y_1,\ldots,y_N)^T$ after a time $t$. This stochastic process is a Dunkl process if $p(t,\by|\bx)$ satisfies
\begin{equation}\label{dunklgeneralbackwardfpe}
\frac{\partial}{\partial t}p(t,\by|\bx)=\frac{1}{2}\sum_{i=1}^N{}_xT_i^2p(t,\by|\bx).
\end{equation}
Note that the first-order derivative and difference terms in \eqref{dunkllaplacian} give the explicit form of the drift and exchange mechanisms, respectively. This means that, in general, Dunkl processes are discontinuous diffusion processes with drift. Note also that if $p(t,\by|\bx)$ is symmetrized with respect to the action of the elements of $W$,
\begin{equation}
\hat{p}(t,\by|\bx)=\sum_{\rho\in W}p(t,\by|\rho\bx),
\end{equation}
the exchange (difference) term in \eqref{dunkllaplacian} vanishes, yielding a continuous process. Henceforth, we will say that functions which are symmetric with respect to the action of the elements of $W$ are ``$W$-invariant.'' These continuous-path processes are called ``radial Dunkl processes,'' \cite{gallardoyor05} and several particular cases have been studied as the eigenvalue processes of matrix-valued models.\cite{demni08A} Radial Dunkl processes on the root system $A_{N-1}$ correspond to Dyson's model \cite{dyson62} when the multiplicity is $k=\beta/2$, and radial Dunkl processes on the root system of type $B_N$, correspond to the square roots of the eigenvalues of the Wishart-Laguerre processes \cite{bru91,konigoconnell01} when the multiplicities are chosen as $k_1=\beta/2$ and $k_2=\beta(2\nu+1)/4$ where $\nu$ is the Bessel index (see, e.g., Ref.~\onlinecite{karatzasshreve91}). For consistency with these processes, we use a renormalized set of multiplicities, chosen as follows. We set $k(\balpha)=\beta\kappa(\balpha)/2$, where $\kappa(\balpha)$ satisfies \eqref{multiplicityfunctioncondition}, while fixing one of the multiplicities so that for at least one root, say $\bxi$, $\kappa(\bxi)=1$. Then, \eqref{dunklgeneralbackwardfpe} becomes
\begin{equation}\label{dunklexplicitbackwardfpe}
\frac{\partial}{\partial t}p(t,\by|\bx)=\frac{1}{2}\Delta_x p(t,\by|\bx)+\frac{\beta}{2}\sum_{\balpha\in R_+}\kappa(\balpha)\Big[\frac{\balpha\cdot\bnabla_x}{\balpha\cdot\bx}p(t,\by|\bx)-\frac{\alpha^2}{2}\frac{1-\sigma_{\balpha}}{(\balpha\cdot\bx)^2}p(t,\by|\bx)\Big].
\end{equation}
With the renormalized multiplicities, we reproduce the factor of $\beta/2$ that appears in Dyson's model and in Wishart-Laguerre processes, and we extend its appearance to Dunkl processes on all other root systems. Then, the parameter $\beta$ is a coefficient of the drift term (first term in the brackets) and the exchange term (second term in the brackets). Thus, it represents the strength of both terms relative to the Laplacian.

The intertwining operator\cite{dunkl91}, denoted henceforth by $V_\beta$, is defined by the following properties: $V_\beta$ is linear, it is normalized so that $V_\beta[1]=1$, it preserves the degree of homogeneous polynomials, and for every analytical function $f(\bx)$ it satisfies the relation
\begin{equation}\label{intertwiningrelation}
T_i V_\beta [f(\bx)] = V_\beta \Big[\frac{\partial}{\partial x_i} f(\bx)\Big].
\end{equation}
Note that one can transform the diffusion equation $\partial/\partial t=\Delta/2$ into \eqref{dunklgeneralbackwardfpe} by applying $V_\beta$ from the left. This means that, if we denote the TPD of a simple diffusion by $p_{\text{BM}}(t,\by|\bx)$, then the function $V_\beta p_{\text{BM}}(t,\by|\bx)$ is a solution of \eqref{dunklexplicitbackwardfpe}.
Using $V_\beta$, one can give a formal expression for the joint eigenfunction of the Dunkl operators $\{T_i\}_{i=1}^N$, known as the ``Dunkl kernel'' $E_\beta(\bx,\by)$. This function satisfies the condition $E_\beta(\bzero,\by)=1$, where $\boldsymbol{0}=(0,\ldots,0)^T$, and the equation
\begin{equation}\label{dunklkerneleigenvalueequation}
T_i E_\beta(\bx,\by)=y_iE_\beta(\bx,\by),\ i=1,\ldots,N.
\end{equation}
Using $V_\beta$ and \eqref{intertwiningrelation}, the Dunkl kernel can be written as
\begin{equation}
E_\beta(\bx,\by)=V_\beta \rme^{\bx\cdot\by}.
\end{equation}
The TPD of a Dunkl process is given by \cite{rosler98}
\begin{equation}\label{TransitionDensityExplicit}
p(t,\by|\bx)=w_\beta\left(\frac{\by}{\sqrt{t}}\right)\frac{\rme^{-(y^2+x^2)/2t}}{c_\beta t^{N/2}}V_\beta \rme^{\bx\cdot\by/t},
\end{equation}
where
\begin{equation}
w_\beta(\bx):=\prod_{\balpha\in R_+}|\balpha\cdot\bx|^{\beta \kappa(\balpha)},
\end{equation}
and
\begin{equation}
c_\beta:=\int_{\RR}\rme^{-x^2/2}w_\beta(\bx)\ud \bx,
\end{equation}
which in several cases is a Selberg integral.\cite{mehta04} 
Because the general form of the intertwining operator is unknown, this expression is formally correct but unknown in most cases. Consequently, the difficulty in calculating quantities derived from $p(t,\by|\bx)$ lies in finding useful explicit expressions for the Dunkl kernel.

The present processes are known to have a stationary state if we scale the variable $\by$ as $\bY=\sqrt{\beta t}\by$ (see, e.g., Ref.~\onlinecite{katoritanemura04}). With this scaling, the process probability distribution is given by
\begin{equation}\label{scaleddensityf}
f(t,\bY):=\int_{\RR^N}(\beta t)^{N/2}p(t,\sqrt{\beta t}\bY|\bx)\mu(\bx)\ud\bx,
\end{equation}
where $\mu(\bx)$ is an arbitrary initial distribution. The expectation of a test function $\phi(\bY)$ is given by
\begin{equation}\label{timedependentexpectationmu}
\langle\phi\rangle_{t}:=\int_{\RR^N}\phi(\bY)f(t,\bY)\ud\bY.
\end{equation}
The steady-state distribution of the process is given by
\begin{equation}\label{steadystatedistribution}
\frac{1}{z_\beta}\exp[-\beta F_R(\bY)],
\end{equation}
where
\begin{equation}\label{definitionFR}
F_R(\bY):=\frac{Y^2}{2}-\frac{1}{\beta}\log w_\beta(\bY)=\frac{Y^2}{2}-\sum_{\balpha\in R_+}\kappa(\balpha)\log |\balpha\cdot\bY|,
\end{equation}
and
\begin{equation}
z_\beta:=\int_{\RR^N}\rme^{-\beta F_R(\bY)}\ud\bY=\frac{c_\beta}{\beta^{(N+\beta\gamma)/2}}.
\end{equation}
Here, we have introduced the sum of renormalized multiplicities
\begin{equation}
\gamma:=\sum_{\balpha\in R_+}\kappa(\balpha).
\end{equation}
Because of the form of the steady-state distribution, the parameter $\beta$ is also understood as the inverse temperature. Rewriting \eqref{scaleddensityf} using \eqref{steadystatedistribution} gives
\begin{equation}\label{scaleddensityfconverging}
f(t,\bY)=\frac{\rme^{-\beta F_R(\bY)}}{z_\beta}\int_{\RR^N}\rme^{-x^2/2t}V_\beta\rme^{\sqrt{\beta/t}\bx\cdot\bY}\mu(\bx)\ud\bx\stackrel{t\to\infty}{\longrightarrow}\frac{\rme^{-\beta F_R(\bY)}}{z_\beta}.
\end{equation}

The function $F_R(\bY)$ is clearly $W$-invariant, and we will show in the Appendix that it is convex for $\bY\in\RR^N$ such that $\bY\cdot\balpha\neq 0$ for all $\balpha\in R$. We will also show that it has $|W|$ minima which can be expressed as $\rho \bos$, $\rho\in W$ and $\bos$ is any particular minimum of $F_R(\bY)$. The minima of $F_R(\bY)$ are known as the peak set \cite{dunkl89B} of $R$ and they are all located at a distance $\sqrt{\gamma}$ from the origin. In view of \eqref{scaleddensityfconverging}, we define the steady-state expectation of $\phi(\bY)$ as
\begin{equation}
\langle\phi\rangle:=\int_{\RR^N}\phi(\bY)\frac{\rme^{-\beta F_R(\bY)}}{z_\beta}\ud\bY.
\end{equation}

Denote the space spanned by the root system $R$ by $\spn(R)$, and let us denote the rank of the root system by $d_R$. The form of \eqref{dunklexplicitbackwardfpe} reveals that if $d_R<N$, then the effect of the drift terms due to the roots $\balpha$ is limited to $\spn(R)$, and the process will behave like a free Brownian motion in the part of $\RR^N$ which is orthogonal to $\spn(R)$. Taking this fact under consideration, we will assume that the initial distribution $\mu(\bx)$ is defined so that 
\begin{equation}\label{conditioninitialdistribution}
\mu(\bx)=0\text{ whenever }\bx\notin\spn(R).
\end{equation}

\section{Asymptotic properties}\label{mainresults}
In this section, we give our two main results and illustrate them using the one-dimensional Dunkl process.

\subsection{Approach to the steady-state ($t\to\infty$)}\label{subsectionsteadystate}

Here, we consider the asymptotic behavior in which $\beta>0$ is fixed and $t$ tends to infinity. We focus on the time-dependent expectation $\langle\phi\rangle_t$ and how it converges to the steady-state expectation $\langle\phi\rangle$. 
We introduce a quantity $\delta$ which denotes the portion of the steady-state distribution that we take into consideration, i.e., the amount of the tail of the distribution which we will ignore. We call it the ``tolerance'' parameter.
For any value of $\delta$, there exists a parameter $r=r(\delta)>0$ such that the relationship
\begin{equation}\label{tuningr}
1-\delta=\int_{Y<r\sqrt{\gamma}}\frac{\rme^{-\beta F_R(\bY)}}{z_\beta}\ud\bY
\end{equation}
is satisfied. Note that the peaks of the distribution $\exp[-\beta F_R(\bY)]/z_\beta$ lie at a distance $\sqrt{\gamma}$ from the origin (see Appendix), meaning that $r(\delta)$ must be larger than 1 to effectively cover the largest contribution of $\exp[-\beta F_R(\bY)]/z_\beta$ to the integral. First, we will consider the case in which the initial distribution is given by a delta function.

\begin{theorem}\label{steadystatetheorem}
Consider the initial distribution $\mu(\bx)=f(t=0,\bx)=\delta^{(N)}(\bx-\bx_0)$ with $\bx_0\in\spn(R)$. The time-dependent expectation of a test function $\phi(\bY)$ at time $t$, $\langle\phi\rangle_{t,\bx_0}$, converges to its steady-state expectation $\langle\phi\rangle$ as
\begin{equation}
\langle\phi\rangle_{t,\bx_0}=\langle\phi\rangle\Big\{1+O\Big[\sqrt{\frac{\beta\gamma}{t}}\frac{ r(\delta) x_0}{(1+\beta \gamma /d_R)}\Big]\Big\}
\end{equation}
for $t\gg x_0^2\max(1/\beta \gamma r(\delta)^2,\beta \gamma r(\delta)^2)$.
\end{theorem}

This theorem is a consequence of the following lemma. The variable $\bx$ can be separated into the component which belongs to $\spn(R)$, $\bx_\parallel$, and the component which is orthogonal to $R$, $\bx_\perp$. If the rank of $R$ is $N$, $\bx=\bx_\parallel$ and $\bx_\perp=\bzero$. 

\begin{lemma}\label{vbetaonlinearfns}
The action of $V_\beta$ on the linear function $f(\bx)=\bx\cdot\by$ is given by
\begin{equation}
V_\beta \bx\cdot\by = \frac{\bx_\parallel\cdot\by_\parallel}{1+\beta\gamma/d_R}+\bx_\perp\cdot\by_\perp.
\end{equation}
\end{lemma}

\begin{remark}
Theorem~\ref{steadystatetheorem} gives the relaxation due to the exchange term in \eqref{dunklexplicitbackwardfpe}. In fact, the first-order correction arises from the expansion
\begin{equation}
\int_{\RR^N}\rme^{-x^2/2t}V_\beta\rme^{\sqrt{\beta/t}\bx\cdot\bY}\mu(\bx)\ud\bx=1+\sqrt{\frac{\beta}{t}}V_\beta\bx_0\cdot\bY+O(x_0^2Y^2/t),
\end{equation}
where $\mu(\bx)=\delta^{(N)}(\bx-\bx_0)$. However, if the initial distribution is $W$-invariant,
\begin{equation}
\mu(\bx)=\frac{1}{|W|}\sum_{\rho\in W} \delta^{(N)}(\bx-\rho\bx_0),
\end{equation}
the first-order correction vanishes due to the sum
\begin{equation}
\sum_{\rho\in W}\rho\bx_0\cdot\bY=0.
\end{equation}
At the same time, the exchange term in \eqref{dunklexplicitbackwardfpe} vanishes when $\mu(\bx)$ is $W$-invariant, and only the drift term drives the relaxation. Therefore, the correction term in Theorem~\ref{steadystatetheorem} is produced only by the exchange mechanism. Consequently, the relaxation due to the drift mechanism is of higher order, namely $O(x_0^2r(\delta)^2\gamma/t)$. This means that the relaxation of due to the drift term is faster than the relaxation due to the exchange term. We will discuss this fact in more detail in after Theorem~\ref{TheoremFreezingLimit} and Lemma~\ref{FreezingLimitDunklKernelFullRank} below.
\end{remark}

The proofs of Theorem~\ref{steadystatetheorem} and Lemma~\ref{vbetaonlinearfns} are given in Section~\ref{preparationssteady}. Note that the denominator of the correction term in Theorem~\ref{steadystatetheorem} comes from Lemma~\ref{vbetaonlinearfns}. Our result can be readily extended to a large class of initial distributions. We assume that $\mu(\bx)$ is Riemann-integrable, and we introduce a monotonically-decreasing function $\tau(x)$, which we call the tail function, such that for some large positive constant $X$, the relationship
\begin{equation}\label{distributiontail}
\tau(x)\geq x^{N-1}\int_{\Omega_N}\mu(\bx)\ud\Omega_N, 
\end{equation}
where $\Omega_N$ is the solid angle in $N$-dimensional Euclidean space, is satisfied when $x>X$. Let the integral of the tail function be denoted by
\begin{equation}\label{tailintegral}
T(y):=\int_y^\infty \tau(x)\ud x.
\end{equation}
Note that, because $\mu(\bx)$ is Riemann-integrable, $T(y)$ is monotonically-decreasing and non-negative for any $y>0$. Then, for any given $\epsilon>0$, there exists a value $C=C(\epsilon)>0$ such that the relationship
\begin{equation}\label{candepsilon}
T(C)\leq\epsilon
\end{equation}
is satisfied. Table~\ref{correctionexponents} gives the form of $T(C)$ for a few types of tail function $\tau(x)$.
\begin{table}[t!]
\begin{tabular}{r|ccc}
\hline
\hline
Form of $\tau(x)$\ &\ 0 for $x\geq C$\ &\ $\rme^{-(x/l)^\xi}$, $\xi,l>0$\ &\ $x^{-(1+\zeta)}$, $\zeta>0$\\
\hline
$T(C)$\ &\ $0$ & $\frac{l}{\xi}(\frac{l}{C})^{\xi-1} \rme^{-(C/l)^\xi}$, for $C/l$ large & $C^{-\zeta}/\zeta$\\
\hline
\hline
\end{tabular}
\caption{Form of the tail integral $T(C)$ given by \eqref{tailintegral}.}\label{correctionexponents}
\end{table}
For the given value of $\epsilon$, the result from Theorem~\ref{steadystatetheorem} yields:
\begin{equation}
\langle\phi\rangle_{t}=\langle\phi\rangle\Big\{1+O\Big[\sqrt{\frac{\beta\gamma}{t}}\frac{ r(\delta) C(\epsilon)}{(1+\beta \gamma /d_R)}\Big]+O(\epsilon)\Big\}.
\end{equation}
We omit the proof, as it only requires the use of the mean value theorem for integrals. 

Let us consider the one-dimensional Dunkl process as an example. The root system is $R=B_1$ and $\gamma=1$, and the two Weyl chambers are the intervals $(-\infty,0)$ and $(0,+\infty)$. The steady-state distribution is given by
\begin{equation}\label{1dsteadystatedistribution}
\frac{\rme^{-\beta F_{B_1}(Y)}}{z_\beta}=\frac{\rme^{-\beta Y^2/2}}{z_\beta}|Y|^\beta.
\end{equation}
In this case, $d_{B_1}=N=1$. The probability density of this type of Dunkl process is one of the few that can be calculated exactly. Denoting the Bessel functions of the second kind by $I_\nu(x)$, it is given by \cite{rosler98, rosler08}
\begin{equation}\label{explicit1dtransition}
p_{B_1}(t,y|x)=\frac{\rme^{-(x^2+y^2)/2t}}{2t}\frac{|y|^\beta}{(xy)^{(\beta-1)/2}}\Big[I_{(\beta+1)/2}\Big(\frac{xy}{t}\Big)+I_{(\beta-1)/2}\Big(\frac{xy}{t}\Big)\Big].
\end{equation}
\begin{figure}[!t]
  \centering
  \subfloat[$t=2$]{\includegraphics[width=0.3\textwidth]{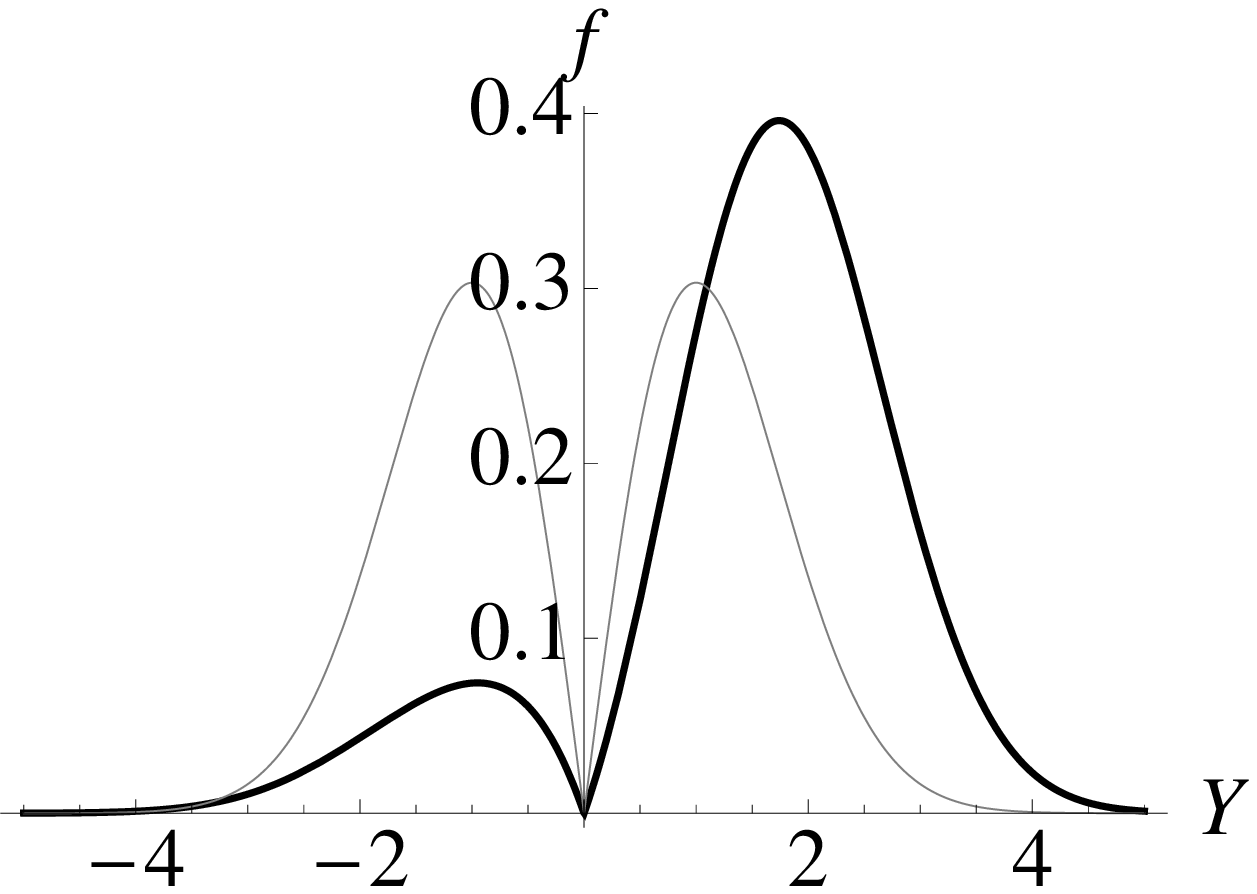}}\                 
  \subfloat[$t=20$]{\includegraphics[width=0.3\textwidth]{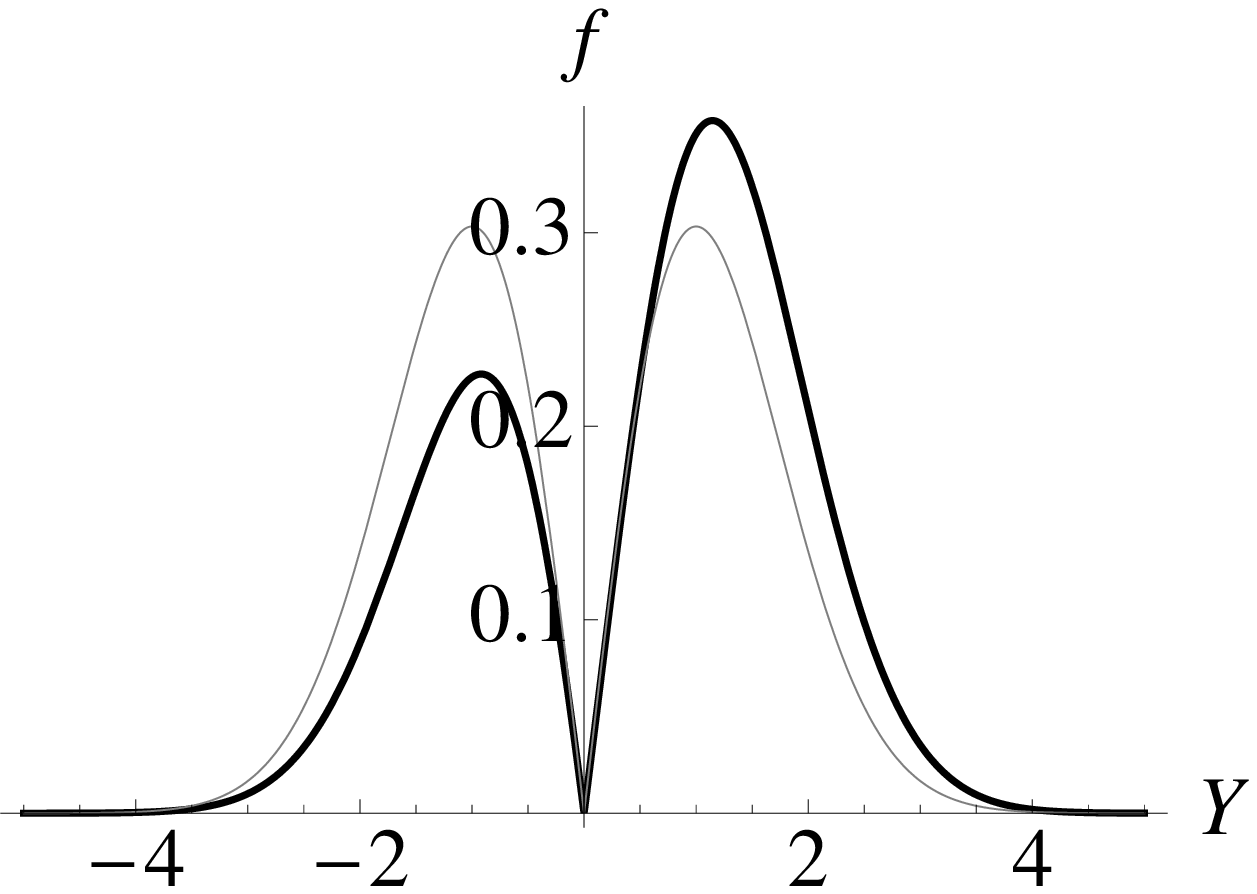}}\\
  
  \subfloat[$t=200$]{\includegraphics[width=0.3\textwidth]{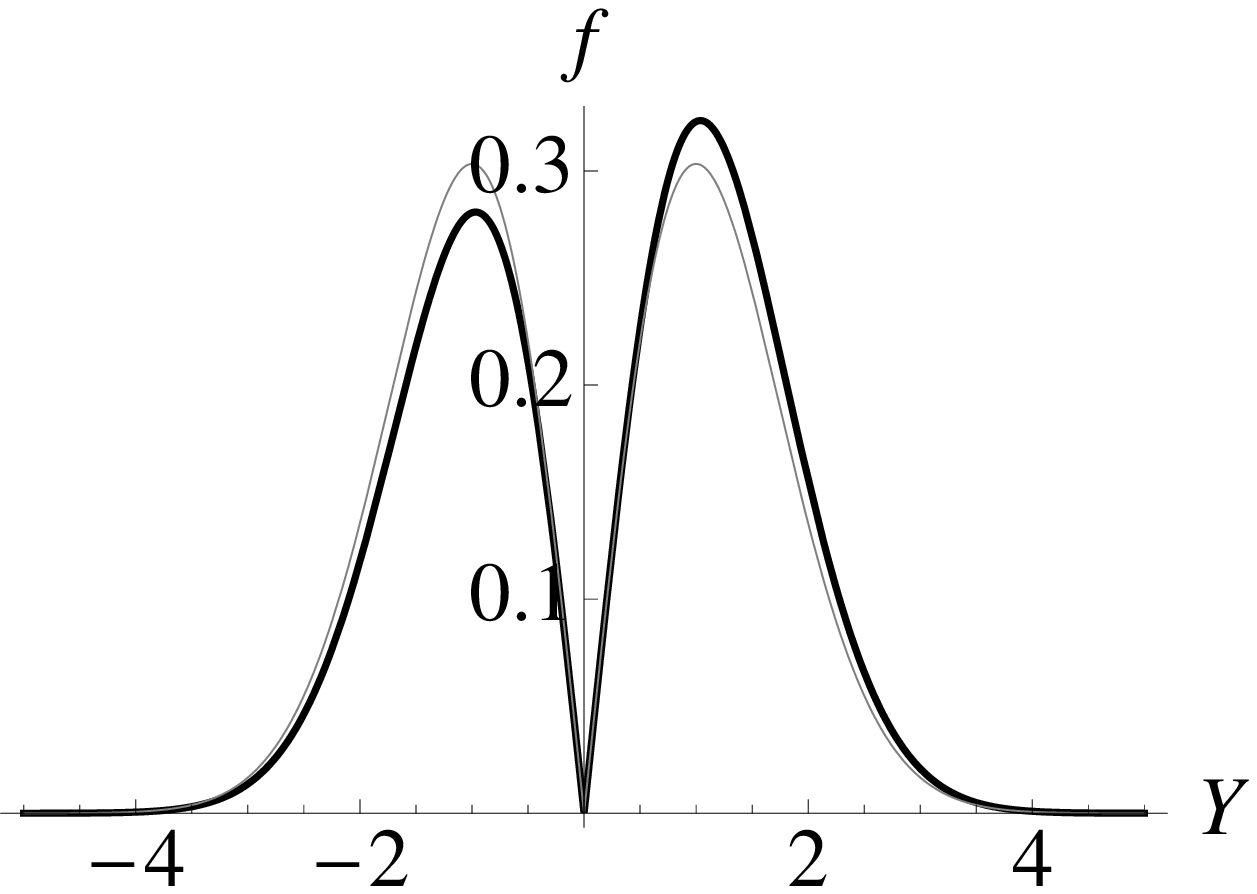}}\
  \subfloat[$t=2000$]{\includegraphics[width=0.3\textwidth]{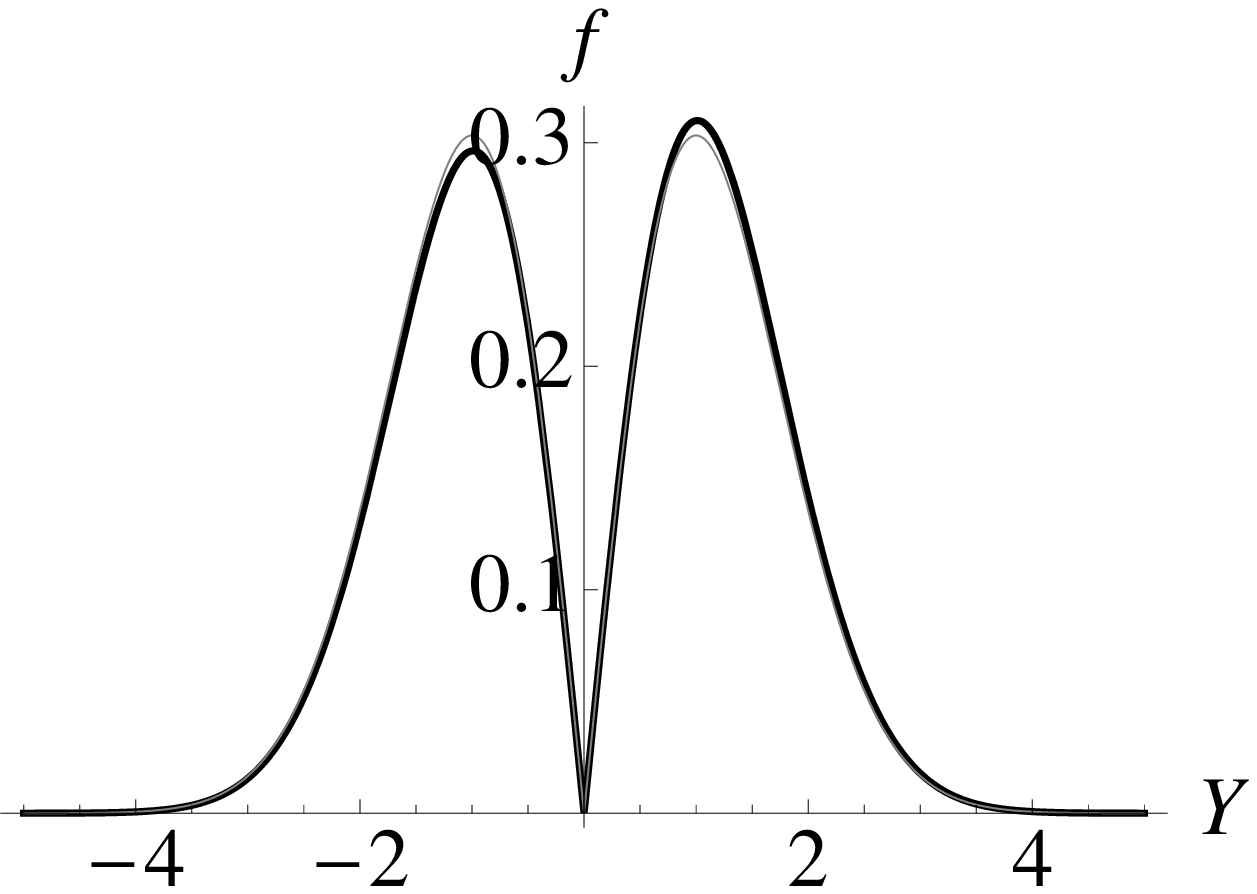}}\\
  \caption{Scaled probability density $f(t,Y)$ (black line) and steady-state probability density (gray line) of the one-dimensional Dunkl process with initial distribution $\mu(x)=\delta(x-2)$ for $\beta=1$ and several values of $t$.} \label{figureonedimensionalrelaxing}
\end{figure}

Figure~\ref{figureonedimensionalrelaxing} depicts the time evolution of the scaled probability density of a one-dimensional Dunkl process with the initial distribution $\mu(x)=\delta(x-x_0)$ with $x_0=2$ and $\beta=1$, and we see that the scaled density converges to the steady-state density as $t$ grows in value. As an example, let us choose $\phi(Y)=Y+1$. Thanks to \eqref{explicit1dtransition}, we can calculate the expectation $\langle \phi \rangle_{t,x_0=2}$ directly,
\begin{equation}
\langle \phi \rangle_{t,x_0=2}=\int_{-\infty}^\infty (Y+1)f(t,Y)\ud Y=1+\frac{2}{\sqrt{t}}.
\end{equation}
Because \eqref{1dsteadystatedistribution} is an even function, it is easy to see that $\langle \phi \rangle=1$. Then, we can write
\begin{equation}
\langle \phi \rangle_{t,x_0=2}=\langle \phi \rangle\Big[1+\frac{2}{\sqrt{t}}\Big]=\langle \phi \rangle[1+O(t^{-1/2})],
\end{equation}
which is consistent with Theorem~\ref{steadystatetheorem}. 

Note that the correction term in Theorem~\ref{steadystatetheorem} depends on $t$, $r(\delta)$ and $x$ in the expected ways: a larger relaxation time is required for large values of $r(\delta)$ and $x$. However, its dependence on $\beta$ is not simple. That is, the correction term is of order $\sqrt{\beta}$ when $\beta$ is small, and it is of order $1/\sqrt{\beta}$ for large $\beta$. Because the correction term at very large values of $\beta$ is small, one may be tempted to take the limit $\beta\to\infty$ from Theorem~\ref{steadystatetheorem}. However, the time required for the theorem to hold is given by $t\gg \beta\gamma x^2 r(\delta)^2$, which tends to infinity in the limit. This means that Theorem~\ref{steadystatetheorem} is not well suited for the strong-coupling limit, and our second result addresses this situation.

\subsection{Approach to the strong-coupling limit ($\beta \to \infty$)}

Here, we consider the case in which $t>0$ is fixed, and $\beta$ tends to infinity. In this regime, we can use a second-order Taylor expansion for $F_R(\bY)$ in order to obtain an approximation of the steady-state distribution function $\exp[-\beta F_R(\bY)]/z_\beta$ using a sum of multivariate Gaussians, which we show in detail in the Appendix. There, we show that the minima of $F_R(\bY)$ occur at the peak set of $R$, which we denote  by $\{\bos_i\}_{i=1}^{|W|}$. It is known that the peak set of the root systems of type $A_{N-1}$ and $B_N$ is given by the zeroes of the Hermite and Laguerre polynomials, which are also known as Fekete points.\cite{deift00} However, we do not expect the peak sets of other root systems to be given by the zeroes of a set of classical orthogonal polynomials in general. The Gaussian approximation of $\exp[-\beta F_R(\bY)]/z_\beta$ is given by
\begin{equation}\label{gaussianapproximationFR}
G_\beta(\bY):=\frac{\beta^{N/2}\sqrt{\det \bH(\bos_1)}}{(2\pi)^{N/2}|W|}\sum_{i=1}^{|W|}\exp[-\beta (\bY-\bos_i)^T \bH(\bos_i) (\bY-\bos_i)/2],
\end{equation}
where we have denoted the Hessian matrix of $F_R(\bY)$ by $\bH(\bY)$ [\eqref{hessianmatrixcomponents} in the Appendix], and we denote the eigenvalues of $\bH$ by $\{\lambda_i\}_{i=1}^{d_R}$. For finite time $t$, we approximate the scaled distribution $f(t,\bY)$ in the same way,
\begin{equation}\label{perturbedgaussianapproximationFR}
\tilde{G}_\beta(\bY):=\frac{\beta^{N/2}\sqrt{\det \tilde{\bH}(\tilde{\bos}_1)}}{(2\pi)^{N/2}|W|}\sum_{i=1}^{|W|}\tilde{c}_i\rme^{-\beta (\bY-\tilde{\bos_i})^T \tilde{\bH}(\tilde{\bos}_i) (\bY-\tilde{\bos}_i)/2}.
\end{equation}
$\tilde{G}_\beta(\bY)$ is a function of the same form as $G_\beta(\bY)$, where the position of the peaks $\{\tilde{\bos}_i\}_{i=1}^{|W|}$, the Hessian matrix $\tilde{\bH}(\bY)$, the eigenvalues $\{\tilde{\lambda}_j\}_{j=1}^{d_R}$, and the coefficients $\{\tilde{c}_i\}_{i=1}^{|W|}$, are time dependent. For the dependence, we have the following theorem:

\begin{theorem}\label{TheoremFreezingLimit}
Consider the initial distribution $\mu(\bx)=\delta^{(N)}(\bx-\bx_0)$ with $\bx_0\in\spn(R)$. 
For $\beta\gg d_R/\gamma$ and $\beta t\gg d_R^2x^2r(\delta)^2/ \gamma$, the time-dependent expectation of a test function $\phi(\bY)$ is approximated by
\begin{equation}\label{largebetatimeexpectation}
\langle\phi\rangle_{t,\bx_0}\approx\int_{\RR^N}\phi(\bY)\tilde{G}_\beta(\bY)\ud\bY,
\end{equation}
where $\tilde{G}_\beta(\bY)$ converges to $G_\beta(\bY)$ in the sense that its peaks lie at 
\begin{equation}
\tilde{\bos}_i(t)=(1+x_0^2/2\gamma\beta t)\bos_i,
\end{equation}
the variances of the Gaussians in the direction of the eigenvectors of $\bH(\bos_i)$ are given by
\begin{equation}
1/\beta\tilde{\lambda}_j(t)=[1+x_0^2/\gamma\beta t]/\beta \lambda_j,
\end{equation}
and the coefficients of the Gaussians are given by
\begin{equation}
\tilde{c}_i(t)=1+\frac{d_R}{\gamma}\frac{\bx_0\cdot\bos_i}{\sqrt{\beta t}}.
\end{equation}
\end{theorem}

In the limit where $\beta\to\infty$, it is easy to see that the scaled probability distribution of a Dunkl process for $t>0$ is given, in the sense of distributions, by
\begin{equation}\label{GeneralFreezingLimit}
\lim_{\beta\to\infty}f(t,\bY)=\frac{1}{|W|}\sum_{i=1}^{|W|}\delta^{(N)}(\bY-\bos_i).
\end{equation}
This equation highlights the fact that when $\beta\to\infty$, the path of the Dunkl process is deterministic, and it is given by the elements of the peak set of $R$. 

Theorem~\ref{TheoremFreezingLimit} depends directly on the following lemma:
\begin{lemma}\label{FreezingLimitDunklKernelFullRank}
For root systems with $d_R=N$, $\beta\gg N/\gamma$ and $N^2x^2y^2/\beta\gamma^2\ll 1$, 
\begin{equation}\label{EquationFreezingLimitKernelFullRank}
V_\beta \rme^{\sqrt{\beta}\bx\cdot\by}\approx \Big(1+\frac{N\bx\cdot\by}{\gamma\sqrt{\beta}}\Big) \exp\Big(\frac{x^2y^2}{2\gamma}\Big).
\end{equation}
\end{lemma}
Indeed, it is due to this exponential form that the perturbation caused by the initial distribution presents itself in $\tilde{G}_\beta(\bY)$ as varying coefficients for each Gaussian, and as a simple power-law correction in the location of the peaks and the variances of the approximating Gaussians. The proofs of Theorem~\ref{TheoremFreezingLimit} and Lemma~\ref{FreezingLimitDunklKernelFullRank} are given in Section~\ref{proofsfreezing}. 

\begin{remark}
Because we have a clearer idea of the form of $f(t,\bY)$ when $\beta$ is large in terms of the location of the Gaussian peaks, their variances and their coefficients, we can isolate the effect of the exchange and drift mechanisms on the function $\tilde G_\beta(\bY)$. Indeed, the effect of the exchange mechanism is found in the coefficients of the Gaussians, which tend to 1 ($\tilde c_i\to1$) as $(\beta t)^{-1/2}$. The correction which appears in the coefficients is dependent upon the product $\bx_0\cdot\bY$, and when the initial distribution is $W$-invariant, these corrections vanish in the same way as the correction term in Theorem~\ref{steadystatetheorem}. Therefore, the effect of the drift mechanism is isolated as the corrections in the shape of $f(t,\bY)$ relative to the approximate steady-state distribution $G_\beta(\bY)$. These corrections are all of order $(\beta t)^{-1}$, which means that if a Dunkl process starts from a non-$W$-invariant initial distribution, the peaks of the distribution $f(t,\bY)$ will settle to their steady-state locations and widths before their heights relax to the same value.
\end{remark}

Theorem~\ref{TheoremFreezingLimit} can be extended to general $\mu(\bx)$ which satisfy condition \eqref{conditioninitialdistribution} in the same way as Theorem~\ref{steadystatetheorem}. Given $\mu(\bx)$ and a parameter $\epsilon>0$, we can find a number $C=C(\epsilon)$ such that \eqref{candepsilon} is satisfied. With $\epsilon$ and $C(\epsilon)$, we have
\begin{equation}
\langle\phi\rangle_{t}=\int_{\RR^N}\phi(\bY)\tilde{G}_\beta(\bY)\ud\bY[1+O(\epsilon)],
\end{equation}
where $\tilde{\bos}_i=(1+C(\epsilon)^2/2\gamma\beta t)\bos_i$, $1/\beta\tilde{\lambda}_j=[1+C(\epsilon)^2/\gamma\beta t]/\beta \lambda_j$ and $\tilde{c}_i=1+d_R \bar{\bx}_\epsilon\cdot\bos_i/\gamma\sqrt{\beta t}$. Here, $\bar{\bx}_\epsilon$ is given by
\begin{equation}
\bar{\bx}_\epsilon=\int_{\bx<C(\epsilon)}\bx\mu(\bx)\ud\bx=O[C(\epsilon)].
\end{equation}

Let us consider the one-dimensional Dunkl process as an example. In this case, the function $F_{B_1}(Y)$ is given by
\begin{equation}
F_{B_1}(Y)=\frac{Y^2}{2}-\log |Y|,
\end{equation}
the peak set is found to be $s=\pm1$, and the second derivative of $F_{B_1}(Y)$ is equal to 2 when $Y=\pm1$. 
We approximate the process density $f(t,Y)$ with the form \eqref{perturbedgaussianapproximationFR}. The result is
\begin{equation}\label{innerapproximationonedimensionlargebeta}
f(t,Y)\approx\frac{1}{2}\sqrt{\frac{\beta\tilde{h}}{2\pi}}(\tilde{c}_{+}\rme^{-\beta\tilde{h}(Y-\tilde{s})^2/2}+\tilde{c}_{-}\rme^{-\beta\tilde{h}(Y+\tilde{s})^2/2}),
\end{equation}
where 
\begin{eqnarray}
\tilde h&=& 2\Big(1-\frac{x_0^2}{\beta t}\Big),\nonumber\\
\tilde s&=&\frac{1}{\sqrt{1-x_0^2/\beta t}}\approx 1+\frac{x_0^2}{2\beta t},\nonumber\\
\tilde{c}_\pm&=&1\pm\frac{x_0}{\sqrt{\beta t}}.
\end{eqnarray}
Clearly, the peak of these Gaussians converges to $\pm1$ with a correction of order $(\beta t)^{-1}$. Similarly, their variance converges to $1/2\beta$ with a correction of order $(\beta t)^{-1}$. However, the coefficients of the Gaussians converge to 1 more slowly, as $(\beta t)^{-1/2}$.

\begin{figure}[!t]
  \subfloat[$\beta=2$]{\includegraphics[width=0.3\textwidth]{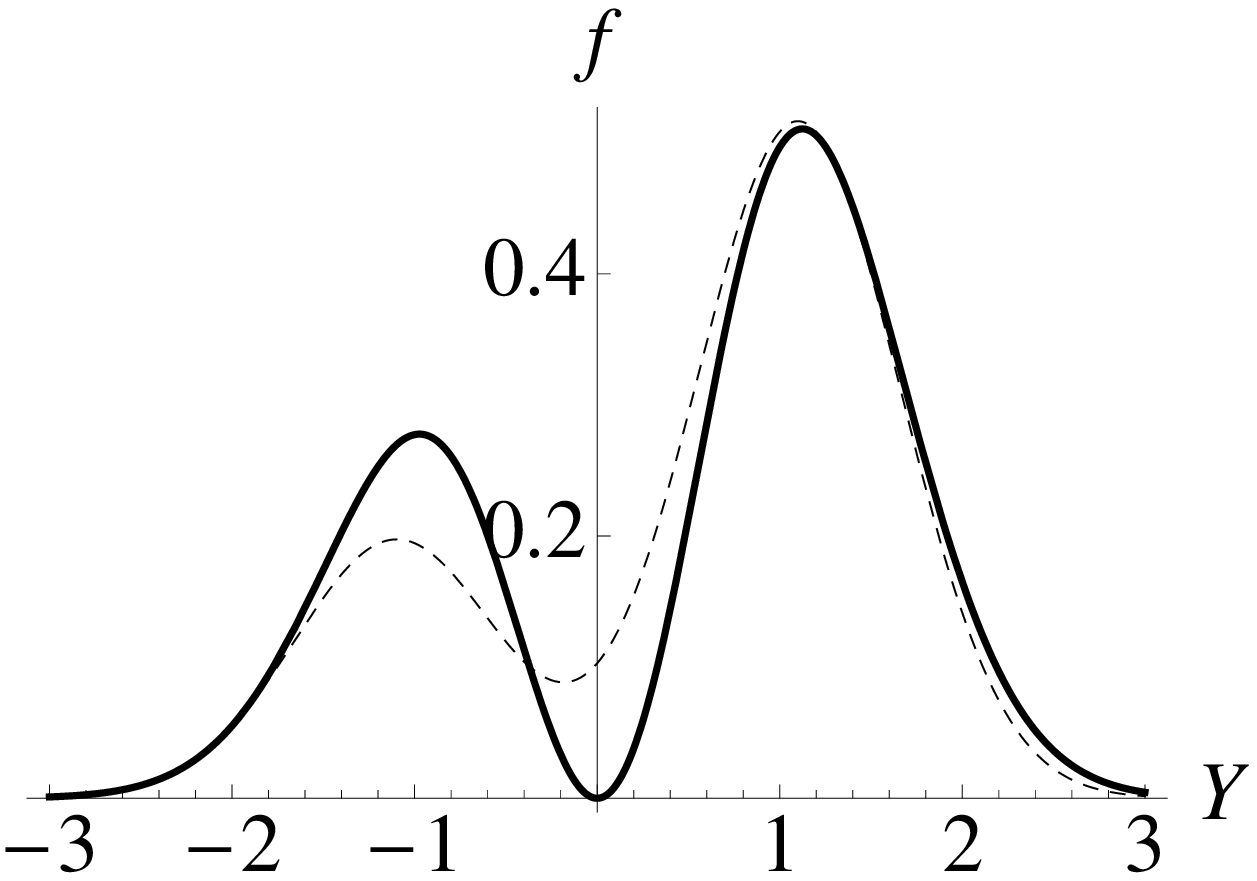}}\                 
  \subfloat[$\beta=100$]{\includegraphics[width=0.3\textwidth]{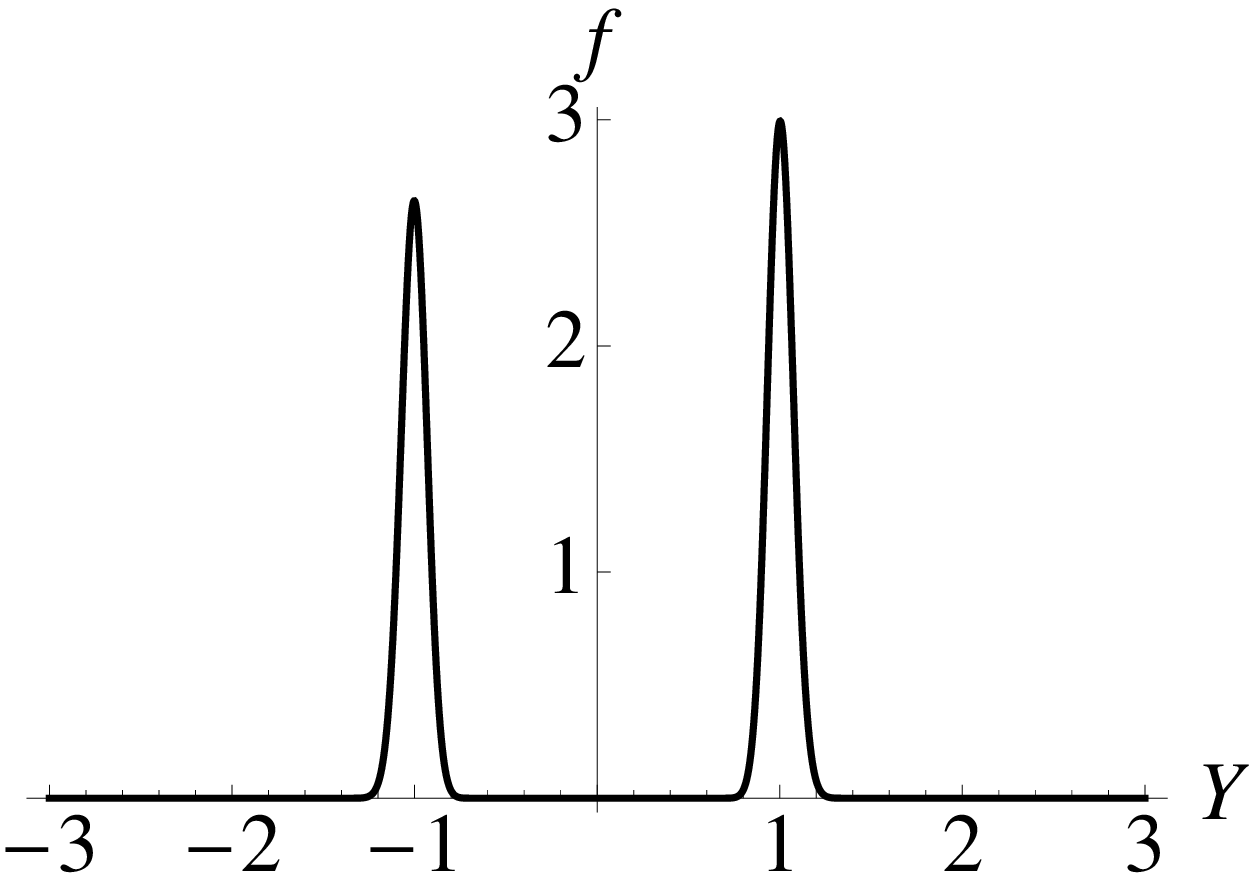}}\
  \subfloat[$\beta=5000$]{\includegraphics[width=0.3\textwidth]{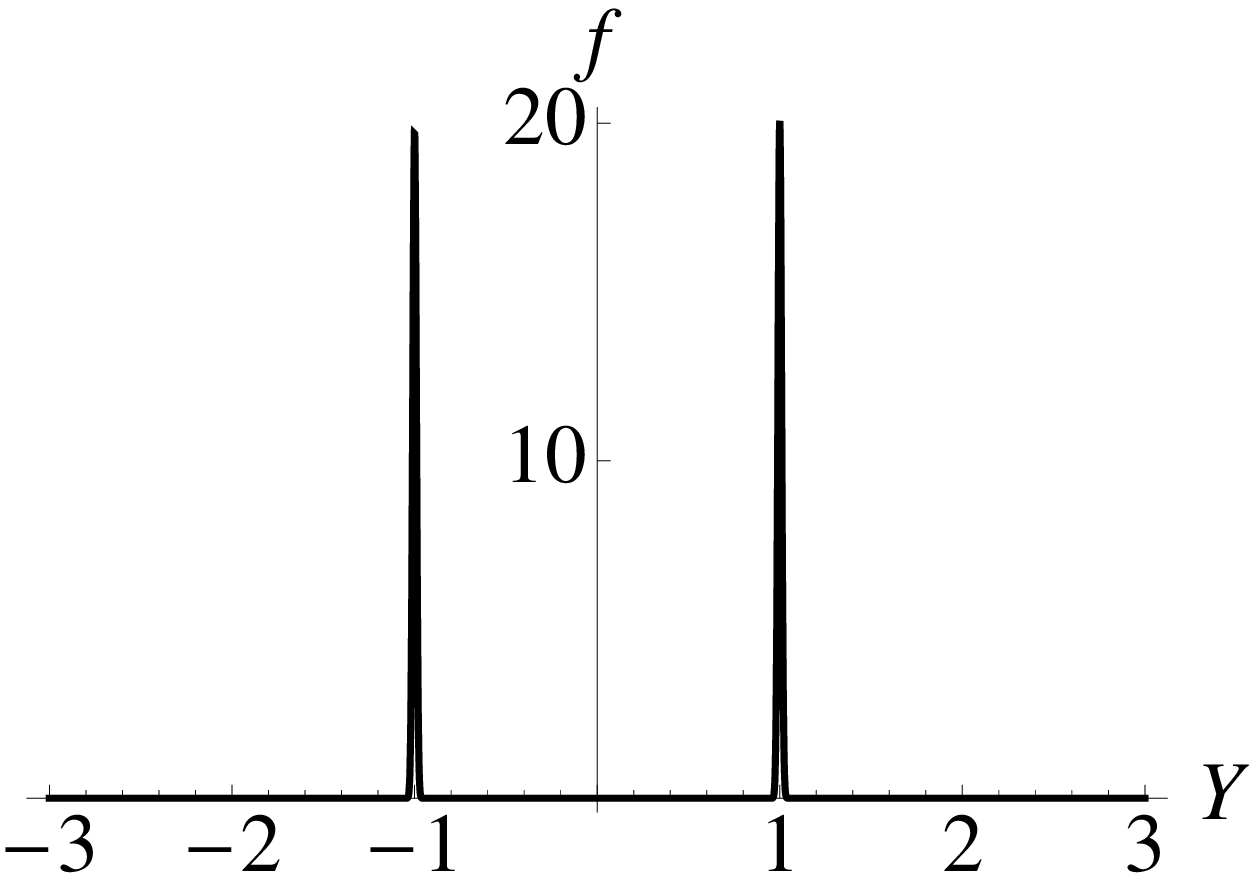}}\\
  \caption{Scaled distribution $f(t,Y)$ (solid line) and $\tilde{G}_\beta (Y)$ given in \eqref{innerapproximationonedimensionlargebeta} (dashed line) for the one-dimensional Dunkl processes with $\mu(x)=\delta(x-2)$ and $t=10$ for various values of $\beta$. Note that the curves are indistinguishable at $\beta\geq100$. As $\beta\to\infty$, both functions tend to a sum of delta functions of equal amplitude located at $Y=\pm 1$.}
  \label{figurefreezingB1}
\end{figure}
We illustrate the approach to the limit $\beta \to \infty$ for the one-dimensional case and the initial distribution $\mu(x)=\delta(x-x_0)$ with $x_0=2$ at $t=10$ in Figure~\ref{figurefreezingB1}. When $\beta=2$, $\tilde{G}_\beta(\bY)$ and $f(t,\bY)$ are clearly different, but when $\beta=100$, the curves appear to fit perfectly well. In addition, at $\beta=100$ the peaks are centered at $Y=\pm 1$, and their width is given by $\sqrt{2\sigma^2}\approx 1/\sqrt{\beta}=0.1$. However, the amplitude of the peaks is still uneven. This is evidence of the fact that the correction due to the drift term in \eqref{dunklexplicitbackwardfpe} is already very small, but the correction due to the exchange term is not. When $\beta=5000$, the peaks have the appearance of delta functions, and most importantly, their amplitudes are almost equal, as we expected.

\begin{figure}[!t]
  \centering
  \subfloat[$t=1$]{\includegraphics[width=0.3\textwidth]{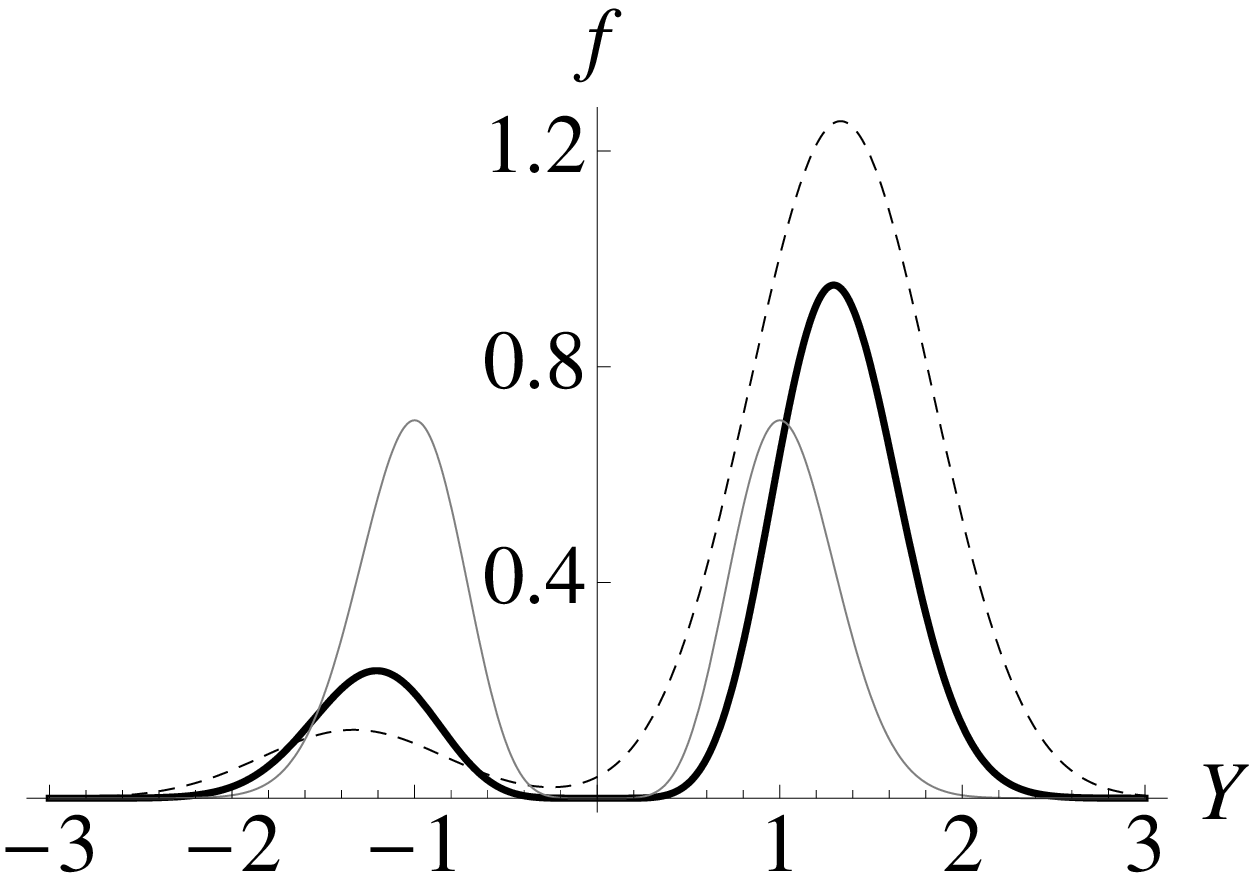}}\                 
  \subfloat[$t=10$]{\includegraphics[width=0.3\textwidth]{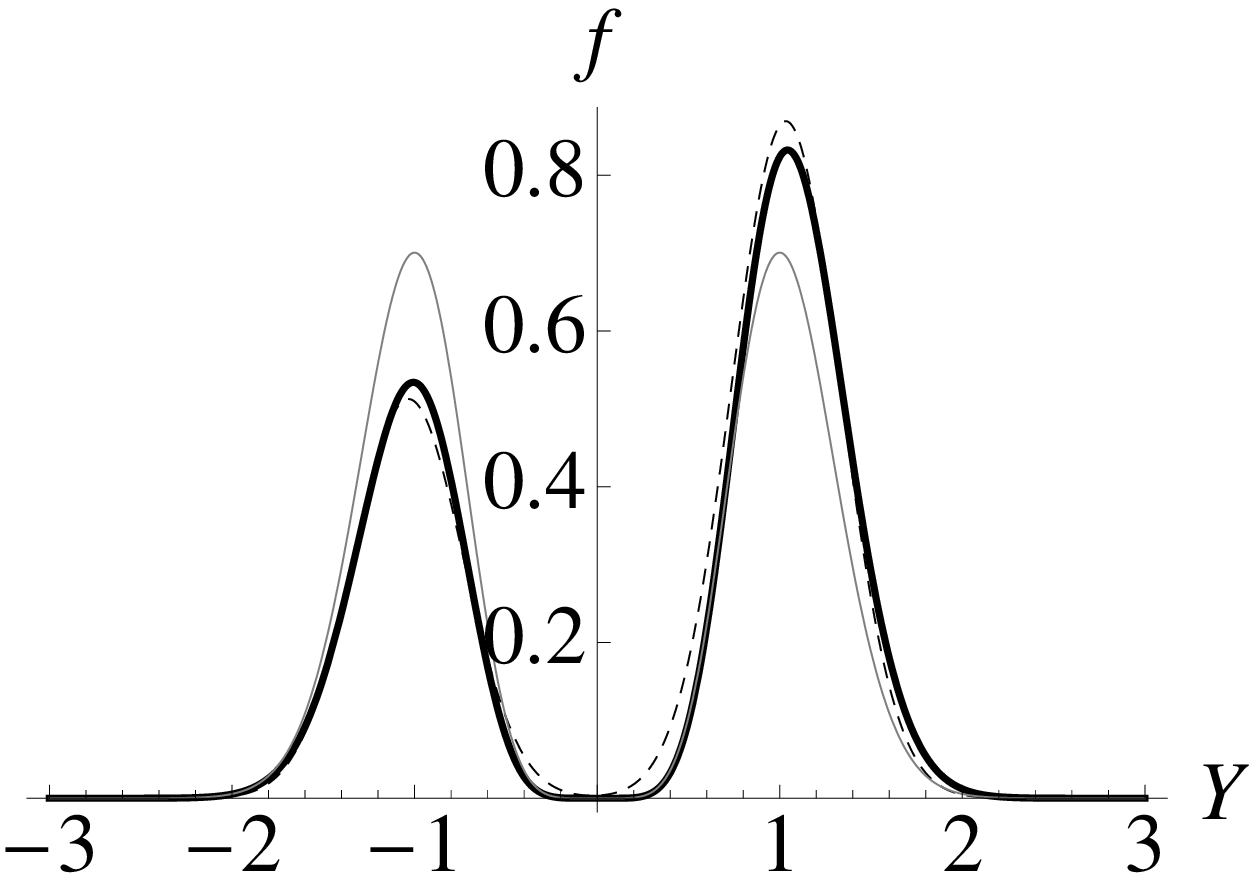}}\\
  \subfloat[$t=100$]{\includegraphics[width=0.3\textwidth]{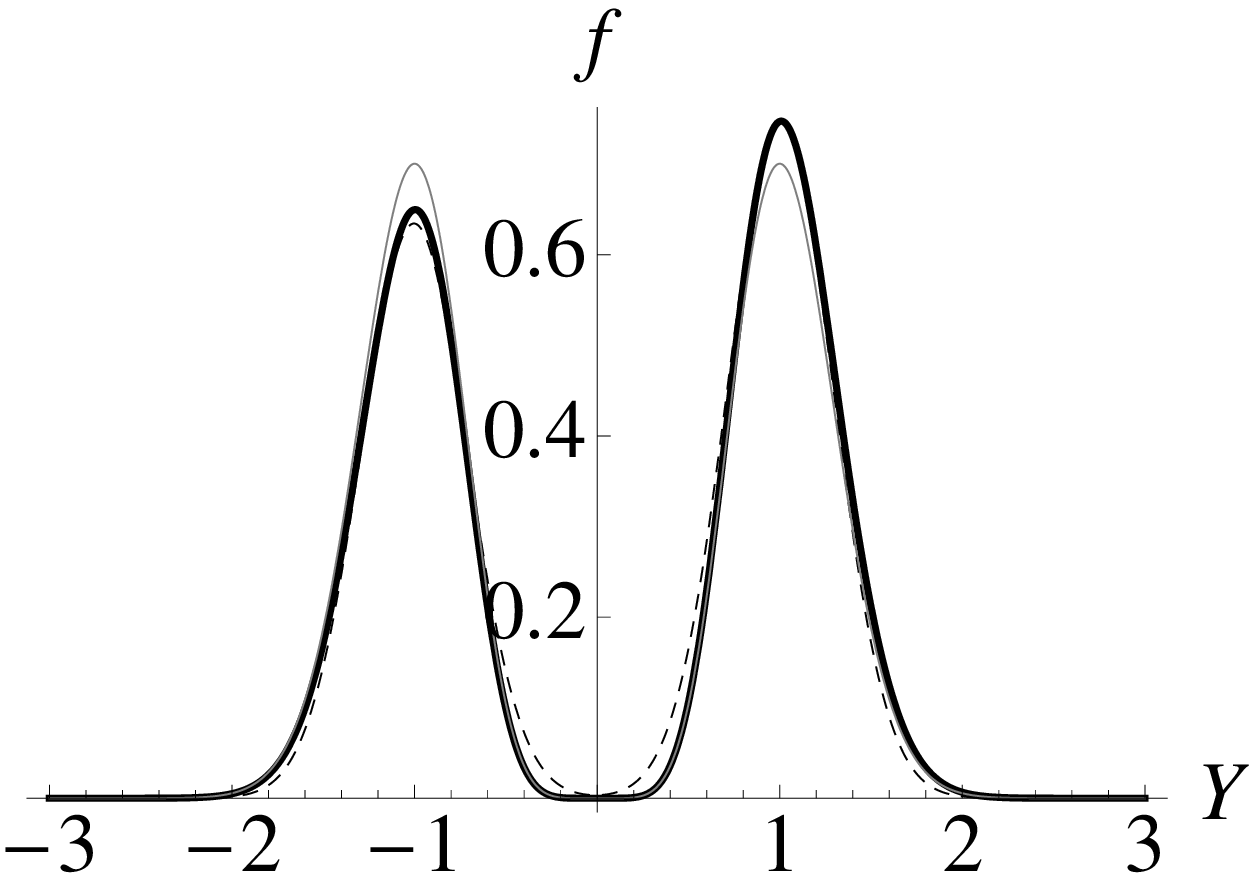}}\
  \subfloat[$t=1000$]{\includegraphics[width=0.3\textwidth]{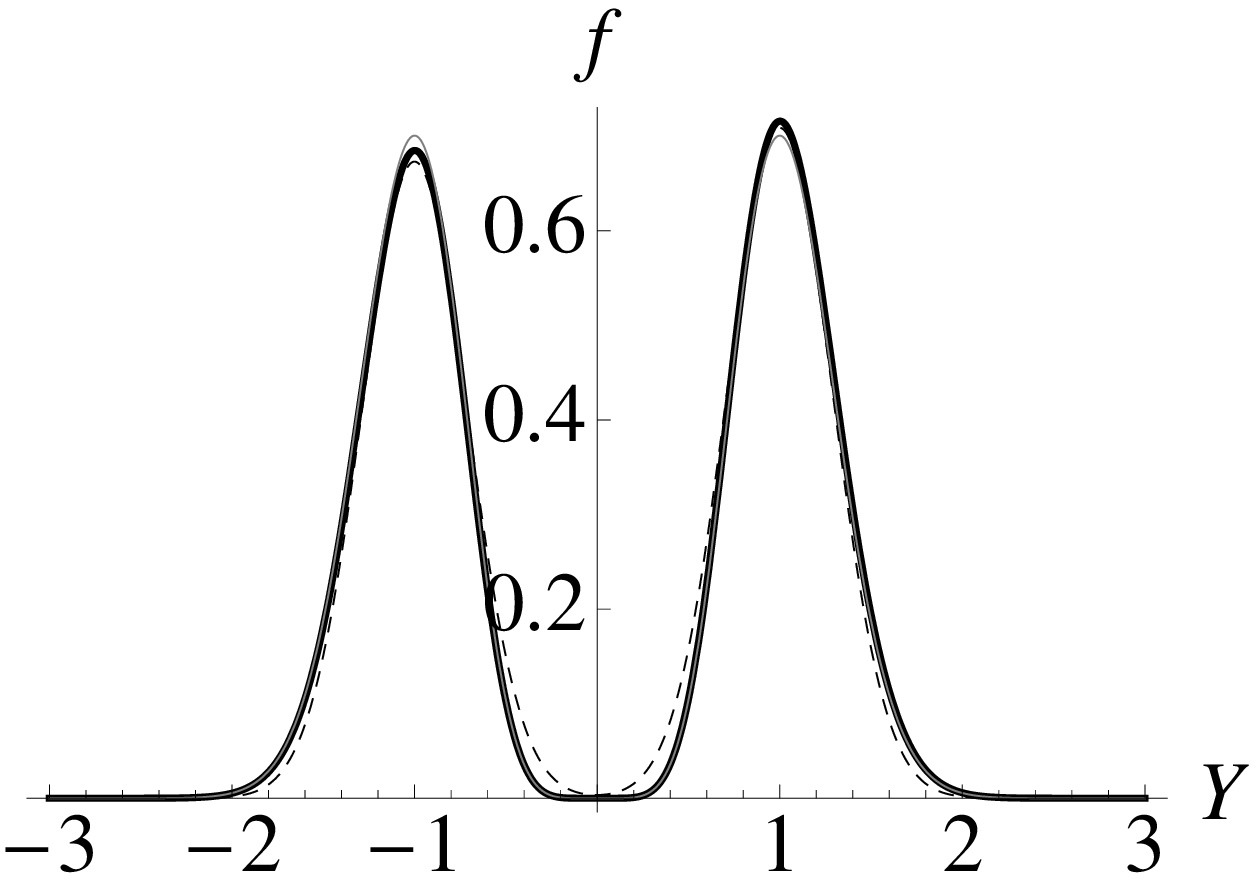}}\\
  \caption{Steady-state distribution $\rme^{-\beta F_{B_1}(Y)}/z_\beta$ (gray line), Gaussian approximation $\tilde G_\beta(Y)$ (dashed line), and scaled distribution $f(t,Y)$ (black line) of the one-dimensional Dunkl process with $\mu(x)=\delta(x-2)$ and $\beta=6$ for varying $t$. As $t\to\infty$, both $f(t,Y)$ and $\tilde G_\beta(Y)$ approach $\rme^{-\beta F_{B_1}(Y)}/z_\beta$.}
  \label{coldsteady}
\end{figure}

Theorem~\ref{TheoremFreezingLimit} also provides information about the convergence to the steady state for large $\beta$. If $\beta$ is taken as a large but fixed quantity and we let $t$ grow, we see that when $t\to\infty$ the approximated distribution $\tilde G_\beta(\bY)$ tends to $G_\beta(\bY)$. We also see that the convergence is actually faster for larger values of $\beta$, as the corrections from $G_\beta(\bY)$ are given by powers of $\beta t$. This is illustrated in Figure~\ref{coldsteady}, where we depict the time evolution of $f(t,Y)$ and $\tilde G_\beta (Y)$ for a one-dimensional Dunkl process with initial distribution $\mu=\delta(x-2)$ at $\beta=6$. We can observe that at $t=10$, $\tilde G_\beta (Y)$ already provides a good approximation of the shape of $f(t,Y)$. We can also observe that both $\tilde G_\beta(\bY)$ and $f(t,Y)$ have peaks that are located as $Y=\pm 1$ and their widths are close to those shown by the steady-state distribution $\rme^{-\beta F_{B_1}(Y)}/z_\beta$, meaning that the relaxation due to the drift mechanism is almost complete. Finally, we see that the relaxation due to the exchange mechanism takes a longer amount of time to complete. Indeed, when $t=100$ the only feature of $f(t,Y)$ that still differs significantly from the steady-state distribution is the height of the probability peaks. In fact, for the case of Figure~\ref{coldsteady}, we require a time of about $t=1000$ in order to have peaks which are equal in height to within 5\%.


\section{Proof of theorems and lemmas}\label{proofs}

In this section, we give proofs of our main results. First, we focus on the approach to the steady state, while the strong-coupling asymptotics is treated in the second part.

\subsection{Proofs of Theorem~\ref{steadystatetheorem} and Lemma~\ref{vbetaonlinearfns}}\label{preparationssteady}

We begin with the results that correspond to the approach to the steady-state, $t\to\infty$. We give the proof of Lemma~\ref{vbetaonlinearfns} first, followed by the proof of Theorem~\ref{steadystatetheorem}. Our proof of Lemma~\ref{vbetaonlinearfns} is based on the procedure outlined in part (5) of Examples~7.1 of Ref.~\onlinecite{roslervoit98}, and extends it to give the effect of the intertwining operator on linear functions in an arbitrary root system $R$.

\begin{proof}[Proof of Lemma~\ref{vbetaonlinearfns}]
Because $V_\beta$ is linear, there exists a real symmetric matrix $\bM_\beta$ such that
\begin{equation}\label{linearintertwiningoperatormatrix}
V_\beta \bx\cdot\by=\bx^T\bM_\beta\by.
\end{equation}
Inserting this relationship in \eqref{intertwiningrelation} with $f(\bx)=\bx\cdot\by$, we obtain
\begin{equation}\label{linearintertwiningrelation1}
y_i=T_i(\bx^T\bM_\beta\by)=[\bM_\beta\by]_i+\frac{\beta}{2}\sum_{\alpha\in R_+}\alpha_i\kappa(\balpha)\frac{(1-\sigma_{\balpha})\bx^T\bM_\beta\by}{\balpha\cdot\bx}.
\end{equation}
At the same time, the difference term can be found to be
\begin{equation}
(1-\sigma_{\balpha})\bx^T\bM_\beta\by=2\frac{\balpha\cdot\bx}{\alpha^2}\balpha^T\bM_\beta\by.
\end{equation}
As the solution of this relation, we obtain
\begin{equation}\label{equationformbeta}
\bM_\beta=\Big(\bI+\beta\sum_{\alpha\in R_+}\kappa(\balpha)\frac{\balpha\balpha^T}{\alpha^2}\Big)^{-1}.
\end{equation}
To calculate $\bM_\beta$, we separate $\bx$ into $\bx_\parallel$ and $\bx_\perp$. For $\bx_\perp$ we have
\begin{equation}
\Big(\bI+\beta\sum_{\alpha\in R_+}\kappa(\balpha)\frac{\balpha\balpha^T}{\alpha^2}\Big)\bx_\perp=\bx_\perp,
\end{equation}
and thus, within the space that is orthogonal to the linear envelope of $R$, $\bM_\beta$ behaves like the identity matrix. For $\bx_\parallel$, i.e., the space spanned by $R$, we rewrite the sum on the r.h.s. of \eqref{equationformbeta} as follows: denote by $n_R$ the number of independent multiplicities for $R$, denote the multiplicities themselves by $\{\kappa_i\}_{i=1}^{n_R}$, and choose roots $\{\bxi_i\}_{i=1}^{n_R}$ such that $\kappa(\bxi_i)=\kappa_i$. Also, define the set $W\bxi_i=\{\rho\bxi_i : \rho\in W\}$. Then, the sum over $R_+$ can be rewritten as
\begin{equation}
\sum_{\balpha\in R_+}\kappa(\balpha)\frac{\balpha\balpha^T}{\alpha^2}=\sum_{i=1}^{n_R}\frac{\kappa_i}{|\bxi_i|^2}\frac{|R_+\cap W\bxi_i|}{|W|}\sum_{\rho\in W}(\rho\bxi_i)(\rho\bxi_i)^T,\label{ProjectorSumUnsolved}
\end{equation}
where the ratio $|R_+\cap W\bxi_i|/|W|$ is included to account for multiple counting on the sum over $\rho$. Because each of the elements of $W$ has a faithful representation in terms of a matrix of size $d_R$, we find that the $jl$th component of the sum is given by
\begin{eqnarray}
\sum_{\rho\in W}[(\rho\bxi_i)(\rho\bxi_i)^T]_{jl}&=&\sum_{n,n^\prime=1}^{d_R}[\bxi_i]_n[\bxi_i]_{n^\prime}\sum_{\rho\in W}[\rho]_{jn}[\rho]_{n^\prime l}\nonumber\\
&=&\frac{|W|}{d_R}\sum_{n,n^\prime=1}^{d_R}[\bxi_i]_n[\bxi_i]_{n^\prime}\delta_{nn^\prime}\delta_{jl}=\frac{|W|}{d_R}|\bxi_i|^2\delta_{jl}.
\end{eqnarray}
Schur's orthogonality relations \cite{sternberg94} allow us to calculate the sum over $\rho$ and obtain the second equality above. Therefore, denoting by $\bI_R$ the identity matrix corresponding to the space spanned by $R$, we obtain
\begin{equation}\label{differencetermonlinearfunctions}
\sum_{\balpha\in R_+}\kappa(\balpha)\frac{\balpha\balpha^T}{\alpha^2}=\sum_{i=1}^{n_R}\kappa_i\frac{|R_+\cap W\bxi_i|}{d_R}\bI_R=\frac{\gamma}{d_R}\bI_R.
\end{equation}
Combining the above results, we have 
\begin{equation}
\Big(\bI+\beta\sum_{\alpha\in R_+}\kappa(\balpha)\frac{\balpha\balpha^T}{\alpha^2}\Big)\bx=\bx_\perp+\Big(1+\beta\frac{\gamma}{d_R}\Big)\bx_\parallel.
\end{equation}
Consequently, the action of $\bM_\beta$ on $\bx$ is found to be
\begin{equation}
\bM_\beta\bx=\bx_\perp+\frac{\bx_\parallel}{1+\beta\gamma/d_R}.
\end{equation}
This last expression, combined with \eqref{linearintertwiningoperatormatrix} completes the proof.
\end{proof}

Having proved Lemma~\ref{vbetaonlinearfns}, we continue with the proof of Theorem~\ref{steadystatetheorem}. For the statements that follow, we recall an important property of the intertwining operator which is a consequence of a theorem by R{\"o}sler (Theorem~1.2 in Ref.~\onlinecite{rosler99}). For any analytical function $f(\bx)$ within the $N$-dimensional ball of radius $|\bx|$, one has the following bound,
\begin{equation}\label{roslersbound}
|V_\beta f(\bx)|\leq \sup_{\by\in\co(W\bx)}|f(\by)|,
\end{equation}
where $\co(W\bx)$ denotes the convex hull of the set $W\bx=\{\bz:{}^\exists \rho\in W, \bz=\rho \bx\}$. In particular, the Dunkl kernel is bounded by
\begin{equation}\label{boundsfordunklkernel}
\rme^{-x y}\leq V_\beta \rme^{\bx\cdot\by}\leq \rme^{x y}.
\end{equation}

\begin{proof}[Proof of Theorem~\ref{steadystatetheorem}] Consider the initial distribution $\mu(\bx)=\delta^{(N)}(\bx-\bx_0)$ with $\bx_0\in\spn (R)$. The corresponding distribution is
\begin{equation}
f(t,\bY)=\frac{\rme^{-\beta F_R(\bY)}}{z_\beta}\rme^{-x_0^2/2t}V_\beta \rme^{\sqrt{\beta/t}\bx_0\cdot\bY}.
\end{equation}
The objective is to find out how the time-dependent expectation $\langle \phi\rangle_{t,\bx_0}$ converges to $\langle \phi\rangle$ as $t$ grows. To evaluate the expectation, we divide the integral over $\bY$ into two regions: $Y<r\sqrt{\gamma}$ and $Y\geq r\sqrt{\gamma}$. The parameter $r(\delta)=r>1$ is obtained using \eqref{tuningr} by choosing the value of $\delta$ so that the integral over $Y<r\sqrt{\gamma}$ covers all the interesting features of the steady-state distribution. The inner part of the integral can be written as
\begin{eqnarray}
\mathcal{I}_{\rmi}&:=&\int_{Y<r\sqrt{\gamma}}\phi(\bY)\frac{\rme^{-\beta Y^2/2}}{z_\beta}w_\beta(\bY)\rme^{-x_0^2/2t}V_\beta\rme^{\sqrt{\beta/t}\bx_0\cdot\bY}\ud \bY\nonumber\\
&\approx&\int_{Y<r\sqrt{\gamma}}\phi(\bY)\frac{\rme^{-\beta Y^2/2}}{z_\beta}w_\beta(\bY)\Big(1+\sqrt{\frac{\beta}{t}}\frac{\bx_0\cdot\bY}{1+\beta\gamma/d_R}\Big)\ud \bY\nonumber\\
&=&\int_{Y<r\sqrt{\gamma}}\phi(\bY)\frac{\rme^{-\beta Y^2/2}}{z_\beta}w_\beta(\bY)\ud \bY\Big(1+O\Big[\sqrt{\frac{\beta \gamma r^2 x_0^2}{t(1+\beta\gamma/d_R)^2}}\Big]\Big).
\end{eqnarray}
For the second line, we used Lemma~\ref{vbetaonlinearfns} and we assumed that $x_0^2/2t\ll \sqrt{\beta / t} \bx_0\cdot\bY\ll 1$ to make the approximation
\begin{equation}
\rme^{-x_0^2/2t}V_\beta \rme^{\sqrt{\beta / t}\bx_0\cdot\bY}\approx 1+\sqrt{\frac{\beta}{t}}\frac{\bx_0\cdot\bY}{1+\beta\gamma/d_R}. 
\end{equation}
This requires the condition $t\gg  x_0^2 \max(\frac{1}{2},\beta\gamma r^2)$. The outer part of the integral,
\begin{equation}
\mathcal{I}_{\text{o}}:=\int_{Y\geq r\sqrt{\gamma}}\phi(\bY)\frac{\rme^{-\beta Y^2/2}}{z_\beta}w_\beta(\bY)\rme^{-x_0^2/2t}V_\beta\rme^{\sqrt{\beta/t}\bx_0\cdot\bY}\ud \bY,
\end{equation}
can be estimated using \eqref{boundsfordunklkernel}:
\begin{equation}
\int_{Y\geq r\sqrt{\gamma}}\phi(\bY)\frac{\rme^{-\beta (Y+x_0/\sqrt{\beta t})^2/2}}{z_\beta}w_\beta(\bY)\ud \bY\leq\mathcal{I}_{\text{o}}\leq \int_{Y\geq r\sqrt{\gamma}}\phi(\bY)\frac{\rme^{-\beta (Y-x_0/\sqrt{\beta t})^2/2}}{z_\beta}w_\beta(\bY)\ud \bY.
\end{equation}
In this inequality, we have assumed that $\phi(\bY)$ is positive (if there are regions where $\phi(\bY)$ is negative, $\mathcal{I}_\text{o}$ can be divided into the regions where $\phi(\bY)$ is positive and the regions where it is negative; in the negative regions, the direction of the inequalities is reversed, and the rest of the argument is unchanged.) We can neglect the effect of $x_0/\sqrt{\beta t}$ by assuming that $t\gg x_0^2/\beta \gamma r^2$ and obtain 
\begin{equation}\label{approximationtailintegral}
\mathcal{I}_{\text{o}}\approx\int_{Y\geq r\sqrt{\gamma}}\phi(\bY)\frac{\rme^{-\beta Y^2/2}}{z_\beta}w_\beta(\bY)\ud \bY.
\end{equation}
Because we can choose $r(\delta)$ large enough (i.e., $\delta$ small enough) to let the inner integral $\mathcal{I}_\rmi$ account for the most significant contribution, we can assume that the approximation error made by neglecting the term $x_0/\sqrt{\beta t}$ in the outer integral is dominated by the correction obtained for the inner integral. Therefore, we write
\begin{equation}
\langle\phi\rangle_{t,\bx_0}=\mathcal{I}_\rmi +\mathcal{I}_{\text{o}}=\langle\phi\rangle\Big[1+O\Big(\sqrt{\frac{\beta \gamma r(\delta)^2 x_0^2}{t(1+\beta\gamma/d_R)^2}}\Big)\Big],
\end{equation}
provided $t\gg x_0^2\max(1/\beta \gamma r(\delta)^2,\beta \gamma r(\delta)^2)$.
\end{proof}

\subsection{Proofs of Theorem~\ref{TheoremFreezingLimit} and Lemma~\ref{FreezingLimitDunklKernelFullRank}}\label{proofsfreezing}

As before, we give the proof of Lemma~\ref{FreezingLimitDunklKernelFullRank} followed by the proof of Theorem~\ref{TheoremFreezingLimit}. However, the proof of Lemma~\ref{FreezingLimitDunklKernelFullRank} requires several other lemmas which we prove first. In particular, we must guarantee the convergence of the limit
\begin{equation}
V_\infty f(\bx):=\lim_{\beta\to\infty}V_\beta f(\bx).
\end{equation}
It has been shown that the action of the intertwining operator on homogeneous polynomials $p(\bx)$ of degree $n$ is given explicitly by\cite{deleavaldemniyoussfi15}
\begin{equation}\label{explicitVbeta}
V_\beta p(\bx)=\sum_{\{g_i\in W\}_{i=1}^n}C(g_1,\ldots,g_n)(g_1\bx\cdot\bnabla)\cdots(g_n\bx\cdot\bnabla)p(\bx),
\end{equation}
where the coefficient $C(g_1,\ldots,g_n)$ is given by
\begin{equation}
C(g_1,\ldots,g_n):=C_n(g_n)C_{n-1}(g_n^{-1}g_{n-1})\cdots C_1(g_2^{-1}g_1),
\end{equation}
each of the factors $C_n(g)$ is given by
\begin{equation}
C_n(g):=\sum_{m=0}^\infty \Big(\frac{\beta}{2}\Big)^m\frac{c_m(g)}{(n+\beta\gamma/2)^{m+1}},
\end{equation}
and the functions $c_m(g)$ are defined by
\begin{equation}
c_m(g):=\sum_{\substack{(\balpha_1,\ldots,\balpha_m)\in R_+^m:\\\sigma_{\balpha_1}\cdots\sigma_{\balpha_m}=g}}\ \prod_{j=1}^m \kappa(\balpha_j)
\end{equation}
for $m\geq 1$ and by $c_0(g)=\delta_{g,1}$ for $m=0$. Note that only the factors $C_n(g)$ depend on $\beta$, and in particular $C_n(g)\sim 1/\beta$, so the sharpest decay for a polynomial of degree $n$ is given by
\begin{equation}\label{betaasymptotics}
V_\beta p(\bx) \sim \beta^{-n},
\end{equation}
meaning that for any ($\beta$-independent) analytical function, $V_\infty f(\bx)=f(\bzero)$. In particular, the sharpest decay of $V_\beta p(\sqrt{\beta}\boldsymbol{x})$ is given by
\begin{equation}\label{scaledbetaasymptotics}
V_\beta p(\sqrt{\beta}\bx)\sim \beta^{-n/2},
\end{equation}
so $V_\beta p(\sqrt{\beta}\boldsymbol{x})$ converges at infinite $\beta$. However, we will see that $V_\beta \exp(\sqrt{\beta}\bx\cdot\by)$ has a non-trivial limit as $\beta\to\infty$. We now prove a statement which will be useful to assert the $W$-invariance of $V_\infty f(\bx)$.
\begin{lemma}\label{nulldifference}
An analytical function $f(\bx)$ is $W$-invariant if and only if it satisfies the equation
\begin{equation}\label{reflectioneigenvalueequation}
\frac{1}{\gamma}\sum_{\balpha\in R_+}\kappa(\balpha)\sigma_{\balpha}f(\bx)=f(\bx).
\end{equation}
\end{lemma}
\begin{proof}
It is clear that if $f(\bx)$ is $W$-invariant, then \eqref{reflectioneigenvalueequation} is satisfied. For the converse, we only need to regard $f(\bx)$ as a homogeneous polynomial of degree $n$. Define the operator
\begin{equation}
Af(\bx):=\frac{1}{2\gamma}\sum_{\balpha\in R}\kappa(\balpha)\sigma_{\balpha}f(\bx)=\frac{1}{\gamma}\sum_{\balpha\in R_+}\kappa(\balpha)\sigma_{\balpha}f(\bx).
\end{equation}
The objective, then, is to prove that the polynomial eigenfunctions of $A$ with eigenvalue 1 are $W$-invariant. It is easy to show that $\rho A=A \rho$ for all $\rho\in W$, and consequently $A$ commutes with the operator
\begin{equation}
Bf(\bx):=\frac{1}{|W|}\sum_{\rho\in W}\rho f(\bx).
\end{equation}
This operator is a projector because $B^2f(\bx)=Bf(\bx)$, and therefore it has two eigenvalues: $0$ and $1$. Because $A$ and $B$ commute, there exists a basis on the space of homogeneous polynomials of degree $n$ such that both operators are diagonalized. Let $p(\bx)$ be an element of that basis. Then, we have either $Bp(\bx)=p(\bx)$ or $Bp(\bx)=0$. The first case indicates that $p(\bx)$ is $W$-invariant, and consequently $Ap(\bx)=p(\bx)$. Therefore, we only need to prove that the non-$W$-invariant eigenfunctions of $A$ (those for which $Bp(\bx)=0$) have eigenvalues different from $1$. In that case, there exists a set $S_+\subseteq W$ for which $\nu p(\bx)\geq 0$ for all $\nu\in S_+$ and $\nu p(\bx)< 0$ for all $\nu\in W\setminus S_+=:S_-$ such that
\begin{equation}
\sum_{\nu\in S_+}\nu p(\bx)=-\sum_{\nu\in S_-}\nu p(\bx).
\end{equation}
Then, we set $Ap(\bx)=\lambda p(\bx)$ and we have
\begin{equation}
\frac{1}{2\gamma}\sum_{\alpha\in R}\kappa(\balpha)\sigma_{\balpha}\sum_{\nu\in S_+}\nu p(\bx)=\frac{1}{2\gamma}\sum_{\nu\in S_+}\nu\sum_{\alpha\in R}\kappa(\balpha)\sigma_{\balpha}p(\bx)=\lambda\sum_{\nu\in S_+}\nu p(\bx),
\end{equation}
from which we obtain
\begin{equation}
0\leq |\lambda|\sum_{\nu\in S_+}\nu p(\bx)=\Big|\frac{1}{2\gamma}\sum_{\alpha\in R}\kappa(\balpha)\sum_{\nu\in S_+}\sigma_{\balpha}\nu p(\bx)\Big|=\Big|\frac{1}{2\gamma}\sum_{\alpha\in R}\kappa(\balpha)\sum_{\nu\in \sigma_{\balpha}S_+}\nu p(\bx)\Big|,
\end{equation}
where we have used the substitution $\sigma_{\balpha}\nu\to\nu$ and $\sigma_{\balpha}S_+=\{\nu\in W : \sigma_{\balpha}\nu\in S_+\}$. Now, we note that the double sum on the right is bounded,
\begin{equation}
\Big|\frac{1}{2\gamma}\sum_{\alpha\in R}\kappa(\balpha)\sum_{\nu\in \sigma_{\balpha}S_+}\nu p(\bx)\Big|\leq\Big|\frac{1}{2\gamma}\sum_{\alpha\in R}\kappa(\balpha)\sum_{\nu\in S_+}\nu p(\bx)\Big|=\Big|\sum_{\nu\in S_+}\nu p(\bx)\Big|=\sum_{\nu\in S_+}\nu p(\bx),
\end{equation}
with equality when $\sigma_{\balpha}S_+=S_+$ for all $\balpha$. This is only possible in two cases. In the first case, $S_+=W$, and so $\nu p(\bx)=0$ for all $\nu$, a $W$-invariant function. In the second case, $S_+=\emptyset$, or $S_-=W$, leading to $\sum_{\nu\in W}\nu p(\bx)<0,$ a contradiction. Therefore, we can write
\begin{equation}
0\leq |\lambda|\sum_{\nu\in S_+}\nu p(\bx)\leq \sum_{\nu\in S_+}\nu p(\bx),
\end{equation}
and we conclude that $|\lambda|\leq 1$, with $|\lambda|=1$ only when $\nu p(\bx)=0$ for all $\nu\in W$. 
\end{proof} 

As a corollary, any function $f(\bx)$ is $W$-invariant if and only if $T_i f(\bx)=\frac{\partial}{\partial x_i}f(\bx)$ for all $i=1,\ldots,N$. However, we use the lemma to prove the following statement about $V_\infty f(\bx)$.

\begin{lemma}\label{winvariance}
Let $f(\bx)$ be an analytical function. Then the function $V_\infty f(\bx)$, if the limit converges, is $W$-invariant.
\end{lemma}
\begin{proof}
Consider the expression $\sum_i x_i T_i V_\beta f(\bx)$. After using \eqref{intertwiningrelation}, we obtain
\begin{equation}
\frac{1}{\beta}\sum_{i=1}^Nx_i\Big[V_\beta\frac{\partial}{\partial x_i}f(\bx)-\frac{\partial}{\partial x_i}V_\beta f(\bx)\Big]=\frac{1}{2}\sum_{\balpha\in R_+}\kappa(\balpha)\Big[V_\beta f(\bx)-V_\beta f(\sigma_{\balpha}\bx)\Big].
\end{equation}
Due to the asymptotics given in~\eqref{betaasymptotics}, if $V_\infty f(\bx)$ converges, so does $V_\infty\frac{\partial}{\partial x_i}f(\bx)$ because $f(\bx)$ is analytic, and can therefore be written as a sum of homogeneous polynomials. Consequently, as $\beta\to\infty$ the l.h.s. vanishes, and we obtain~\eqref{reflectioneigenvalueequation}. By Lemma~\ref{nulldifference}, it follows that $V_\infty f(\bx)$ is $W$-invariant. 
\end{proof}

We turn our attention to the limit $V_\infty \rme^{\bx\cdot\by}$. Because $\rme^{\bx\cdot\by}$ is an analytical function, $V_\infty \rme^{\bx\cdot\by}$ converges and it is a $W$-invariant function. Recall that the Dunkl kernel satisfies \eqref{dunklkerneleigenvalueequation}, but as $\beta$ tends to infinity, we will need a first-order operator which preserves $W$-invariance in order to calculate $V_\infty \rme^{\bx\cdot\by}$ explicitly. It is known that the Dunkl operators are $W$-equivariant\cite{dunklxu,rosler08}, so if $f(\bx)$ is $W$-invariant, then
\begin{equation}
\rho \Big[\sum_{i=1}^N(\xi_i T_i)\Big]f(\bx)=\sum_{i=1}^N[(\rho\bxi)_i T_i)]\rho f(\bx)=\sum_{i=1}^N[(\rho\bxi)_i T_i)]f(\bx)
\end{equation}
for $\bxi\in\RR^N$. If we want the operator $\sum_{i=1}^N\xi_i T_i$ to preserve the $W$-invariance of $f(\bx)$, we require $\rho \bxi=\bxi$ for all $\rho\in W$, meaning that $\bxi$ must be orthogonal to $\spn(R)$. Consequently, we can only have first order Dunkl operators which preserve $W$-invariance if $d_R<N$. 

On the other hand, if $d_R=N$, we can use the Dunkl Laplacian, which preserves $W$-invariance for any root system.\cite{dunklxu,rosler08} This means that we can use the equation
\begin{equation}\label{dunklkernelequationlaplacian}
\sum_{i=1}^N T_i^2 V_\beta \rme^{\bx\cdot\by}=y^2V_\beta\rme^{\bx\cdot\by}
\end{equation}
to calculate $V_\beta \rme^{\sqrt{\beta}\bx\cdot\by}$ as $\beta\to\infty$. With these facts in mind, we can prove the following.

\begin{lemma}\label{FreezingLimitDunklKernelNonFullRank}
For root systems with $d_R<N$, the limit $\beta\to\infty$ of the Dunkl kernel is given by
\begin{equation}\label{EquationFreezingLimitDunklKernelNonFullRank}
V_\infty\rme^{\bx\cdot\by}=\rme^{\bx_\perp\cdot\by_\perp}.
\end{equation}
\end{lemma}

\begin{proof}
For this derivation, denote $V_\infty\rme^{\bx\cdot\by}$ by $g(\bx,\by)$. By Lemma~\ref{winvariance}, the function $g(\bx,\by)$ must be $W$-invariant. At the same time, \eqref{dunklkerneleigenvalueequation} must hold at finite $\beta$. However, the operator $T_{\bxi}=\sum_{i=1}^N\xi_i T_i$ does not preserve $W$-invariance unless $\bxi$ is orthogonal to $\spn(R)$. Therefore, the equation
\begin{equation}
T_{\bxi} V_\beta \rme^{\bx\cdot\by}=\bxi\cdot\by V_\beta \rme^{\bx\cdot\by}
\end{equation}
only holds in the limit $\beta\to\infty$ when $\bxi$ is orthogonal to $R$, otherwise it must be zero because $W$-invariant and non-$W$-invariant functions cannot be identically equal.

Suppose that the space orthogonal to $R$ has an orthonormal basis denoted by $\{\bphi_i\}_{i=1}^{N-d_R}$. Then, for $1\leq i\leq N-d_R$, one has
\begin{equation}
T_{\bphi_i}=\bphi_i\cdot\bnabla+\frac{\beta}{2}\sum_{\balpha\in R_+}[\bphi_i\cdot\balpha] \kappa(\balpha) \frac{1-\sigma_{\balpha}}{\balpha\cdot\bx}=\bphi_i\cdot\bnabla,
\end{equation}
and when $\beta\to\infty$, 
\begin{equation}\label{DunklKernelFirstOrderDifferentialEquation}
\bphi_i\cdot\bnabla g(\bx,\by)=[\bphi_i\cdot\by] g(\bx,\by).
\end{equation}
Note that if $\bxi$ is not a linear combination of the $\{\bphi_i\}_{i=1}^{N-d_R}$, then $T_{\bxi} g(\bx,\by)=0$. Because $\bphi_i\cdot\by$ is the $i$th component of $\by$ in the space orthogonal to $R$, and \eqref{DunklKernelFirstOrderDifferentialEquation} holds for $1\leq i\leq N-d_R$, it follows that
the solution of \eqref{DunklKernelFirstOrderDifferentialEquation} is $g(\bx,\by)=\rme^{\bx_\perp\cdot\by_\perp}.$
\end{proof}

If $d_R=N$, it follows immediately from this result that $V_\infty \rme^{\bx\cdot\by}=1$. However, we are interested in the limit when $\beta\to\infty$ of $V_\beta \rme^{\sqrt{\beta}\bx\cdot\by}$. Note that the $W$-invariant part of the Dunkl kernel, known as the generalized Bessel function,
\begin{equation}
E_\beta^W(\bx,\by):=\frac{1}{|W|}\sum_{\rho\in W}V_\beta\exp(\rho\bx\cdot\by),
\end{equation}
decays more slowly with growing $\beta$ than the asymptotics given in~\eqref{betaasymptotics}.\cite{deleavaldemniyoussfi15} In fact, the $n$th term in the homogeneous polynomial expansion of $E_\beta^W(\bx,\by)$ is given by
\begin{equation}
E_{\beta,n}^W(\bx,\by)\propto\sum_{\{g_i\in W\}_{i=1}^n}C_{n-1}(g_n^{-1}g_{n-1})\cdots C_1(g_2^{-1}g_1)\prod_{j=1}^n(g_j\bx\cdot\by),
\end{equation}
with $E_{\beta,0}^W(\bx,\by)=1$, and because each factor of $C_{j-1}(g_j^{-1}g_{j-1})$ contributes a factor of $\beta^{-1}$ it follows that the maximum decay of $E_{\beta,n}^W(\bx,\by)$ is  $\beta^{-(n-1)}$. Note that the linear term vanishes because $\sum_{g\in W}g\bx\cdot\by=0$. Therefore, the constant and linear terms are independent of $\beta$ and $\bx$, and if we replace $\bx$ with $\sqrt{\beta}\bx$, for $n\geq 2$ we have a maximum decay of $\beta^{1-n/2}$ for the $n$th order term. This means that $V_\beta\exp(\sqrt{\beta}\bx\cdot\by)$ should converge to a second-degree polynomial at $\beta\to\infty$ if its maximum decay is its actual decay. However, as we will show below, the decay of each term in the expansion of $V_\beta\exp(\bx\cdot\by)$ is weaker, giving a non-trivial limit for the scaled Dunkl kernel $V_\beta\exp(\sqrt{\beta}\bx\cdot\by)$.


\begin{proof}[Proof of Lemma~\ref{FreezingLimitDunklKernelFullRank}]
We begin by deriving the decay with $\beta$ of each of the terms in the expansion
\begin{equation}\label{DunklKernelExpansion}
V_\beta \rme^{\bx\cdot\by}=\sum_{n=0}^\infty V_\beta\frac{(\bx\cdot\by)^n}{n!}.
\end{equation}
Recall that $V_\beta 1=1$. By Lemma~\ref{vbetaonlinearfns}, the first-order term is
\begin{equation}\label{DunklKernelLinearTerm}
V_\beta\bx\cdot\by=\frac{\bx\cdot\by}{1+\beta\gamma/N}\stackrel{\beta\text{ large}}{\approx}\frac{N\bx\cdot\by}{\beta\gamma}\sim\frac{1}{\beta}.
\end{equation}
By Lemma~\ref{winvariance}, the limit $\beta\to\infty$ eliminates the non-$W$-invariant part of $V_\beta\exp(\bx\cdot\by)$ faster than its $W$-invariant part. Consequently, the slowest decay for each of the terms in \eqref{DunklKernelExpansion} is obtained by using the Dunkl Laplacian, which relates higher-order terms with lower-order terms while conserving their $W$-invariance (or lack thereof).

In general, each term in the expansion~\eqref{DunklKernelExpansion} satisfies the equation
\begin{equation}\label{EquationSetupFrozenDunklKernelFullRank}
\frac{y^2}{\beta}V_\beta\frac{(\bx\cdot\by)^{n-2}}{(n-2)!}=\Big[\frac{1}{\beta}\Delta+\sum_{\balpha\in R_+}\kappa(\balpha)\Big(\frac{\balpha\cdot\bnabla}{\balpha\cdot\bx}-\frac{\alpha^2}{2}\frac{1-\sigma_{\balpha}}{(\balpha\cdot\bx)^2}\Big)\Big]V_\beta\frac{(\bx\cdot\by)^{n}}{n!}
\end{equation}
for $n>1$. We proceed using mathematical induction. Assume that
\begin{equation}\label{DunklKernelBetaAsymptotics}
V_\beta\frac{(\bx\cdot\by)^{2m}}{(2m)!}\sim\frac{1}{\beta^m}\quad\text{and}\quad V_\beta\frac{(\bx\cdot\by)^{2m+1}}{(2m+1)!}\sim\frac{1}{\beta^{m+1}},
\end{equation}
and note that these assumptions hold for $m=0$. Because spatial partial derivatives and $\sigma_{\balpha}$ do not have an effect on the $\beta$-dependence of $V_\beta (\bx\cdot\by)^n$, one may write
\begin{eqnarray}
\sum_{\balpha\in R_+}\kappa(\balpha)\Big(\frac{\balpha\cdot\bnabla}{\balpha\cdot\bx}-\frac{\alpha^2}{2}\frac{1-\sigma_{\balpha}}{(\balpha\cdot\bx)^2}\Big)V_\beta\frac{(\bx\cdot\by)^{n}}{n!}&=&\nonumber\\
\frac{1}{\beta}\Big[y^2V_\beta\frac{(\bx\cdot\by)^{n-2}}{(n-2)!}-\Delta V_\beta\frac{(\bx\cdot\by)^{n}}{n!}\Big]&\stackrel{\beta\text{ large}}{\sim}&
\begin{cases}
\frac{1}{\beta^{m+1}}&\ \text{for }n=2(m+1),\\
\frac{1}{\beta^{m+2}}&\ \text{for }n=2(m+1)+1.
\end{cases}
\end{eqnarray}
Here, we have used the fact that, after being deformed by $V_\beta$, $n$th degree polynomials decay faster than (or at least at the same rate as) $(n-2)$th degree polynomials with growing $\beta$, which is clear from~\eqref{betaasymptotics}. By induction, \eqref{DunklKernelBetaAsymptotics} holds for $m\geq 0$. Then, it follows that
\begin{equation}
V_\beta\frac{\beta^m(\bx\cdot\by)^{2m}}{(2m)!}
\end{equation}
converges to a non-zero, $W$-invariant polynomial as $\beta\to\infty$ and that 
\begin{equation}
V_\beta\frac{\beta^{m+1/2}(\bx\cdot\by)^{2m+1}}{(2m+1)!}\sim\frac{1}{\sqrt{\beta}}\stackrel{\beta\to\infty}{\longrightarrow}0.
\end{equation}
Define the limit of the scaled even terms of the expansion~\eqref{DunklKernelExpansion} by
\begin{equation}
L_m(\bx,\by):=\lim_{\beta\to\infty}V_\beta\frac{\beta^{m}(\bx\cdot\by)^{2m}}{(2m)!}.
\end{equation}
By Lemma~\ref{winvariance}, these functions are $W$-invariant. Multiplying \eqref{EquationSetupFrozenDunklKernelFullRank} by $\beta^{m}$ with $n=2m$ gives
\begin{equation}
y^2 V_\beta\frac{\beta^{m-1}(\bx\cdot\by)^{2(m-1)}}{(2(m-1))!}=\Big[\frac{1}{\beta}\Delta+\sum_{\balpha\in R_+}\kappa(\balpha)\Big(\frac{\balpha\cdot\bnabla}{\balpha\cdot\bx}-\frac{\alpha^2}{2}\frac{1-\sigma_{\balpha}}{(\balpha\cdot\bx)^2}\Big)\Big]V_\beta\frac{\beta^{m}(\bx\cdot\by)^{2m}}{(2m)!}.
\end{equation}
Taking the limit $\beta\to\infty$ gives
\begin{equation}\label{DifferentialEquationFrozenDunklKernelFullRankExpansion}
y^2L_{m-1}(\bx,\by)=\sum_{\balpha\in R_+}\kappa(\balpha)\frac{\balpha\cdot\bnabla L_{m}(\bx,\by)}{\balpha\cdot\bx}.
\end{equation}
This equation has the boundary condition
\begin{equation}\label{EquationBoundaryConditionFrozenDunklKernelFullRank}
L_m(\bzero,\by)=\delta_{0,m}.
\end{equation}
Let us assume the following solution, which satisfies the boundary condition \eqref{EquationBoundaryConditionFrozenDunklKernelFullRank},
\begin{equation}
L_m(\bx,\by)=\frac{1}{m!}\Big(\frac{x^2y^2}{2\gamma}\Big)^m.
\end{equation}
Inserting this form into \eqref{DifferentialEquationFrozenDunklKernelFullRankExpansion} gives
\begin{equation}
\sum_{\balpha\in R_+}\kappa(\balpha)\frac{\balpha\cdot\bnabla L_{m}(\bx,\by)}{\balpha\cdot\bx}=L_{m-1}(\bx,\by)\frac{y^2}{\gamma}\sum_{\balpha\in R_+}\kappa(\balpha)=y^2L_{m-1}(\bx,\by)
\end{equation}
for all $m>0$. Thus, summing up over $m$ we have the limit
\begin{equation}\label{DunklKernelCompleteLimit}
\lim_{\beta\to\infty}V_\beta\rme^{\sqrt{\beta}\bx\cdot\by}=\sum_{m=0}^\infty L_m(\bx,\by)=\exp\Big(\frac{x^2y^2}{2\gamma}\Big).
\end{equation}
Now, we formulate an approximation for the Dunkl kernel for the case where $\beta$ is very large but finite. From our derivation of \eqref{DunklKernelCompleteLimit}, we know that the first-order correction decays with $\beta$ as $\beta^{-1/2}$. From this consideration, we assume the simplest possible form,
\begin{equation}
V_\beta \rme^{\sqrt{\beta}\bx\cdot\by}\approx D(\bx,\by):=\rme^{x^2 y^2/2\gamma}(1+a\bx\cdot\by),
\end{equation}
where $a=a(\beta)$ is determined using \eqref{dunklkerneleigenvalueequation}. Calculating $T_i D(\bx,\by)$ yields
\begin{equation}
T_i D(\bx,\by)=x_i\frac{y^2}{\gamma}D(\bx,\by)+a y_i \rme^{x^2 y^2/2\gamma}+a\frac{\beta}{2}\rme^{x^2 y^2/2\gamma}\sum_{\balpha\in R_+}\alpha_i\kappa(\balpha)\frac{[1-\sigma_{\balpha}]\bx\cdot\by}{\balpha\cdot\bx}.
\end{equation}
From \eqref{differencetermonlinearfunctions}, we find that
\begin{equation}
\frac{\beta}{2}\sum_{\balpha\in R_+}\alpha_i\kappa(\balpha)\frac{[1-\sigma_{\balpha}]\bx\cdot\by}{\balpha\cdot\bx}=\frac{\beta\gamma}{N}y_i,
\end{equation}
so we have
\begin{equation}
T_i D(\bx,\by)=\Big[x_i\frac{y^2}{\gamma}(1+a\bx\cdot\by)+a y_i+a\frac{\beta\gamma}{N}y_i\Big]\rme^{x^2 y^2/2\gamma}.
\end{equation}
We impose the condition $T_i D(\bx,\by)\to\sqrt{\beta}y_iD(\bx,\by)$ for $\beta$ tending to infinity. This yields
\begin{equation}
\Big[x_i\frac{y^2}{\gamma}(1+a\bx\cdot\by)+a y_i+a\frac{\beta\gamma}{N}y_i\Big]/(1+a\bx\cdot\by)\to\sqrt{\beta}y_i,
\end{equation}
meaning that $a=N/\gamma\sqrt{\beta}$ provided that $\beta\gg N/\gamma$, and
\begin{equation}
V_\beta \rme^{\sqrt{\beta}\bx\cdot\by}\approx D(\bx,\by)=\rme^{x^2 y^2/2\gamma}\Big(1+\frac{N\bx\cdot\by}{\gamma\sqrt{\beta}}\Big).
\end{equation}
Finally, because we have approximated the anisotropic part of $V_\beta \rme^{\sqrt{\beta}\bx\cdot\by}$ to first order, this expression holds for $N^2 x^2 y^2/\beta \gamma^2\ll 1$.
\end{proof}

As a direct consequence of Lemmas~\ref{FreezingLimitDunklKernelFullRank} and \ref{FreezingLimitDunklKernelNonFullRank}, we can write an explicit form for the Dunkl kernel for large but finite $\beta$ in any root system.
\begin{corollary}\label{largebetadunklkernel}
The Dunkl kernel can be approximated by
\begin{equation}
V_\beta \rme^{\sqrt{\beta}\bx\cdot\by}\approx \Big(1+\frac{d_R\bx_\parallel\cdot\by_\parallel}{\gamma\sqrt{\beta}}\Big) \exp\Big[\sqrt{\beta}\bx_\perp\cdot\by_\perp+\frac{x_\parallel^2 y_\parallel^2}{2\gamma}\Big]
\end{equation}
in the case where $\beta\gg d_R/\gamma$ and $d_R^2x_\parallel^2 y_\parallel^2/\beta\gamma^2\ll 1$.
\end{corollary}
\begin{proof}
When $d_R=N$, the statement is identical to Lemma~\ref{FreezingLimitDunklKernelFullRank}. When $d_R<N$, one can separate \eqref{dunklkernelequationlaplacian} into the part that corresponds to $\spn(R)$ and the part orthogonal to $R$. The first part obeys Lemma~\ref{FreezingLimitDunklKernelFullRank}, and the second part obeys Lemma~\ref{FreezingLimitDunklKernelNonFullRank}. The product of the two functions yields the result. 
\end{proof}

In principle, we should use this corollary to prove Theorem~\ref{TheoremFreezingLimit}, but imposing the condition \eqref{conditioninitialdistribution} allows us to ignore $\bx_\perp$ and $\bY_\perp$. Therefore, we can use Lemma~\ref{FreezingLimitDunklKernelFullRank} (replacing $N$ by $d_R$) to give the proof of Theorem~\ref{TheoremFreezingLimit}.
\begin{proof}[Proof of Theorem~\ref{TheoremFreezingLimit}]
As in the proof of Theorem~\ref{steadystatetheorem}, we consider $\bx_0\in\spn(R)$. 
Let us rewrite the expectation of $\phi(\bY)$ as
\begin{equation}\label{expectationlargebetafullrank}
\langle\phi\rangle_{t,\bx_0}=\int_{\RR^N}\phi(\bY)\frac{\rme^{-\beta Y^2/2}}{z_\beta}w_\beta(\bY)\rme^{-x_0^2/2t}V_\beta\rme^{\sqrt{\beta/t}\bx_0\cdot\bY}\ud\bY.
\end{equation}

Let us evaluate the inner and outer integrals $\mathcal{I}_\rmi$ and $\mathcal{I}_{\text{o}}$. Using Lemma~\ref{FreezingLimitDunklKernelFullRank}, and assuming that $\beta \gg N/\gamma$, the inner integral is rewritten as 
\begin{eqnarray}
\mathcal{I}_\rmi&=&\int_{Y<r\sqrt{\gamma}}\phi(\bY)\frac{\rme^{-\beta Y^2/2}}{z_\beta}w_\beta(\bY)\rme^{-x_0^2/2t}\Big[1+\frac{d_R}{\gamma}\frac{\bx_0\cdot\bY}{\sqrt{\beta t}}\Big]\rme^{x_0^2 Y^2/2\gamma t}\ud \bY\nonumber\\
&\approx&\int_{Y<r\sqrt{\gamma}}\phi(\bY)\frac{\rme^{-\beta \tilde{F}_R(\bY)}}{z_\beta}\Big[1+\frac{d_R}{\gamma}\frac{\bx_0\cdot\bY}{\sqrt{\beta t}}\Big]\ud \bY,
\end{eqnarray}
where
\begin{equation}
\tilde{F}_R(\bY):=\Big(1-\frac{x_0^2}{\gamma \beta t}\Big)\frac{Y^2}{2}-\sum_{\balpha\in R_+}\kappa(\balpha)\log|\balpha\cdot\bY|+\frac{x_0^2}{2\beta t}.
\end{equation}
We ensure that we can use Lemma~\ref{FreezingLimitDunklKernelFullRank} in the region $Y<r\sqrt{\gamma}$ by imposing the condition $\beta t\gg d_R^2x_0^2r^2/ \gamma$, which implies that $d_R^2x_0^2 Y^2/\beta\gamma^2 t\ll 1$. We can use a second-order approximation for $\tilde{F}_R(\bY)$ to obtain a Gaussian approximation similar to the one obtained in the Appendix. In this case, the minima are given by the vectors $\tilde{\bos}$ which satisfy
\begin{equation}
\sqrt{1-x_0^2/\gamma\beta t}\tilde{\bos}=\frac{1}{\sqrt{1-x_0^2/\gamma\beta t}}\sum_{\balpha\in R_+}\frac{\kappa(\balpha)\balpha}{\balpha\cdot\tilde{\bos}}.
\end{equation}
Setting $\bos=\sqrt{1-x_0^2/\gamma\beta t}\tilde{\bos}$ yields the equation which defines the peak set of $R$, meaning that the minima of $\tilde{F}_R(\bY)$ are located at $\tilde{\bos}_i=\bos_i/\sqrt{1-x_0^2/\gamma\beta t}$, where $\{\bos_i\}_{i=1}^{|W|}$ denotes the peak set. The Hessian matrix of $\tilde{F}_R(\bY)$ evaluated at $\tilde{\bos}_l$ is given by
\begin{equation}
[\tilde{\bH}(\tilde{\bos}_l)]_{ij}:=\frac{\partial^2}{\partial Y_j\partial Y_i}\tilde{F}_R(\bY)\Big|_{\bY=\tilde{\bos}_l}=\Big[1-\frac{x_0^2}{\gamma \beta t}\Big]\Big[\delta_{ij}+\sum_{\balpha\in R_+}\frac{\kappa(\balpha)}{(\balpha\cdot\bos_l)^2}\alpha_i\alpha_j\Big]=\Big[1-\frac{x^2}{\gamma \beta t}\Big][{\bH}({\bos}_l)]_{ij}.
\end{equation}
With these relations, we can write
\begin{equation}
\frac{\rme^{-\beta \tilde{F}_R(\bY)}}{z_\beta}\approx\frac{\beta^{N/2}\sqrt{\det \tilde{\bH}(\tilde{\bos}_1)}}{(2\pi)^{N/2}|W|}\sum_{i=1}^{|W|}\exp[-\beta (\bY-\tilde{\bos_i})^T \tilde{\bH}(\tilde{\bos}_i) (\bY-\tilde{\bos}_i)/2],
\end{equation}
and from the expressions obtained for $\{\tilde{\bos}_i\}_{i=1}^{|W|}$ and $\tilde{\bH}(\tilde{\bos}_l)$, we see that the peaks of $\tilde{G}_\beta(\bY)$ converge to $\{{\bos}_i\}_{i=1}^{|W|}$ as
\begin{equation}
\tilde{\bos}_i\approx\Big(1+\frac{x_0^2}{2\gamma\beta t}\Big){\bos}_i,
\end{equation}
while the variances along the eigenvectors of $\tilde{\bH}(\tilde{\bos}_l)$ are given by
\begin{equation}
\frac{1}{\beta \tilde{\lambda}_i}=\Big[(\beta \lambda_i)\Big(1-\frac{x_0^2}{\gamma \beta t}\Big)\Big]^{-1}\approx\Big(1+\frac{x_0^2}{\gamma \beta t}\Big)/\beta\lambda_i.
\end{equation}
By the mean value theorem for integrals, there exists a set of vectors $\{\tilde{\bu}_i\}_{i=1}^{|W|}$ such that
\begin{equation}
\mathcal{I}_\rmi\approx\frac{\beta^{N/2}\sqrt{\det \tilde{\bH}(\tilde{\bos}_1)}}{(2\pi)^{N/2}|W|}\sum_{i=1}^{|W|}\int_{Y<r\sqrt{\gamma}}\phi(\bY)\Big[1+\frac{d_R}{\gamma}\frac{\bx_0\cdot\tilde{\bu}_i}{\sqrt{\beta t}}\Big]\rme^{-\beta (\bY-\tilde{\bos_i})^T \tilde{\bH}(\tilde{\bos}_i) (\bY-\tilde{\bos}_i)/2}\ud \bY.
\end{equation}
Because $\beta$ is very large, we can assume that the value of $\tilde{\bu}_i$ is very close to $\tilde{\bos}_i$, meaning that we can rewrite the inner integral as
\begin{equation}
\mathcal{I}_\rmi\approx\int_{Y<r\sqrt{\gamma}}\phi(\bY)\tilde{G}_\beta(\bY)\ud \bY,
\end{equation}
and the coefficients of the Gaussians are
\begin{equation}
\tilde{c}_i=1+\frac{d_R}{\gamma}\frac{\bx_0\cdot\tilde{\bos}_i}{\sqrt{\beta t}}\approx1+\frac{d_R}{\gamma}\frac{\bx_0\cdot\bos_i}{\sqrt{\beta t}}.
\end{equation}
The outer integral is treated as in \eqref{approximationtailintegral}, provided $\beta t\gg x_0^2/ \gamma r^2$; this condition is justified by the previous assumption that $\beta t\gg d_R^2x_0^2 r^2/\gamma$, for which $r>1$, and by $d_R\geq 1$. This means that in the region $Y\geq r\sqrt{\gamma}$, the location of the peaks and the width of the Gaussians is perturbed by a maximum amount of order $x_0/\sqrt{\beta t}$. The parameter $r$ can be chosen large enough to make the contribution of the integral $\mathcal{I}_\text{o}$ negligible, as the tail of the steady-state distribution decays like a Gaussian distribution. This means that the expectation is approximately given by the integral $\mathcal{I}_\rmi$, and the distribution of the process can be approximated by $\tilde G_\beta (\bY)$. 
\end{proof}


\section{Concluding remarks and discussion}\label{conclusions}

We obtained two results which describe the behavior of scaled Dunkl processes when they approach the steady state and the strong-coupling limits. As a property of the process approaching the steady state (Theorem~\ref{steadystatetheorem}), we proved that the deviation from the steady-state distribution $\exp[-\beta F_R(\bY)]/z_\beta$ is given by a decay law which depends mainly on the action of the intertwining operator on linear functions. This confirms our previous conjecture \cite{andrauskatorimiyashita14} that the convergence to the steady state should be valid for any value of $\beta$, not necessarily large. Moreover, our result implies that Dunkl processes of type $A_{N-1}$ and type $B_N$ need not be radial to converge to the eigenvalue distributions of the $\beta$-Hermite and $\beta$-Laguerre ensembles of random matrices respectively.

As a property of the strong-coupling limit (Theorem~\ref{TheoremFreezingLimit}), we showed that the scaled distribution of the process can be approximated with the sum of multivariate Gaussians given in \eqref{perturbedgaussianapproximationFR}. We obtained the conditions for which this approximation is valid, and our strong-coupling limit asymptotics are consistent with the Gaussian approximations given for the $\beta$-Hermite and $\beta$-Laguerre eigenvalue distributions in Ref.~\onlinecite{dumitriuedelman05}. We also showed that for $t>0$ the scaled probability distribution converges to a sum of delta functions as $\beta\to\infty$. The delta functions are located at the peak set of the root system under consideration. E.g., for the root systems of type $A_{N-1}$ and $B_N$, these peak sets are given by the zeroes of the Hermite and Laguerre polynomials respectively, which is consistent with our previous results. However, peak sets are not expected to be related to the roots of a set of known orthogonal polynomials in general.

We also found the relationship between the corrections to the steady-state distribution and their corresponding mechanisms. In the approach to the steady state, the first-order correction decays as $t^{-1/2}$, and it is due to the exchange mechanism. When the effect of the exchange is removed by choosing a $W$-invariant initial distribution, the dominating correction decays as $t^{-1}$, which is driven by the drift mechanism. While we found a clear dependence on $\beta$ for the exchange correction, we do not know the exact dependence on $\beta$ of the correction due to the drift mechanism. This dependence must be calculated from the effect of $V_\beta$ on quadratic functions.

In the approach to the strong-coupling limit, we used similar arguments to distinguish the corrections due to the exchange and drift mechanisms. We showed that the exchange corrections are of order $(\beta t)^{-1/2}$ and have an effect on the height of the approximating Gaussians. The drift corrections perturb the shape of the Gaussians, i.e., their location and width, and they are of order $(\beta t)^{-1}$. 

From a more mathematical point of view, the large-$\beta$ asymptotics presented here are based on the $\beta$-dependence of each of the terms in the homogeneous polynomial expansion of the Dunkl kernel. This dependence has been shown to be, at most, of the order of $\beta^{-n}$ for the $n$th degree polynomial,\cite{deleavaldemniyoussfi15} and we have found that this decay is weaker, of order $\beta^{-\lfloor (n+1)/2\rfloor}$. We believe that this must be due to the fact that the Dunkl kernel is the simultaneous eigenfunction of not only Dunkl operators, but of the Dunkl Laplacian as well. Because of the symmetry found in root systems, the term of order $\beta^2$ that one would expect to find in the Dunkl Laplacian for being a second order operator vanishes\cite{dunklxu}, and this is the main reason why we found in the proof of Lemma~\ref{FreezingLimitDunklKernelFullRank} that the $2m$th and $(2m-1)$th degree terms in the Dunkl kernel decay in the same form. This means that there must be a way to show that out of the $n$ terms $C_{j-1}(g_j^{-1}g_{j-1})$ in~\eqref{explicitVbeta}, $\lfloor n/2\rfloor$ terms can be shown to not depend on $\beta$. We do not know at the moment how to prove this, but there is some evidence suggesting that this conjecture may be true, such as the form of the rank-one intertwining operator, the form of the Dunkl kernel for dihedral groups given in Ref.~\onlinecite{deleavaldemniyoussfi15}, and the limit form of the (scaled) generalized Bessel function of type $B_N$ at infinite $\beta$.\cite{andrauskatorimiyashita14}

While we are able to calculate the deviations from the steady-state and strong-coupling limits of the scaled distribution of Dunkl processes, there are several quantities that cannot be calculated using the techniques shown here. In particular, the calculation of the steady-state expectation of $\phi(\bY)$ involves the calculation of integrals of the form
\begin{equation}
\int_{\RR^N}\phi(\bY) \rme^{-\beta Y^2/2}\prod_{\balpha\in R_+}|\balpha\cdot\bY|^{\beta\kappa(\balpha)}\ud\bY,
\end{equation}
which are, in general, not trivial. Perhaps this expectation can be calculated using the Dunkl transform, \cite{rosler08}
\begin{equation}
\hat{f}(\bxi):=\frac{1}{c_\beta}\int_{\RR^N}f(\bY)V_\beta\rme^{-\ii\bY\cdot\bxi}\prod_{\balpha\in R_+}|\balpha\cdot\bY|^{\beta\kappa(\balpha)}\ud\bY,
\end{equation}
where $\ii^2=-1$. Indeed, if we set $\varphi(\bY):=\rme^{-\beta Y^2/2}\phi(\bY)$, then $\hat{\varphi}(\bzero)\propto \langle\phi\rangle$. However, this relationship is of little use in practice because the Dunkl kernel is the integral kernel of the transform, meaning that the calculation of the transform depends on the explicit form of the Dunkl kernel. We would like to investigate the problem further, however, because the calculation of both $\langle\phi\rangle$ and $\langle\phi\rangle_t$ should provide the means to study other aspects of Dunkl processes such as multi-time and single-time correlations.

\begin{acknowledgments}
The authors would like to acknowledge the comments and suggestions of the referee, which greatly helped improve this paper.
SA was supported by the Photon Science Center of the University of Tokyo in the duration of this work. 
SA would like to thank E. Paquette and D. Bananni for stimulating discussions.
\end{acknowledgments}
 
\appendix

\section{Peak Sets}

An important part of the proof of Theorem~\ref{TheoremFreezingLimit} concerns the peak sets introduced by Dunkl \cite{dunkl89B} and the minima of the function $F_R(\bY)$. The extrema of $F_R(\bY)$ occur at the solutions of
\begin{equation}\label{ExtremaPotentialR}
\frac{\partial}{\partial Y_i}F_R(\bY)=Y_i-\sum_{\balpha\in R_+}\frac{\kappa(\balpha)}{\balpha\cdot\bY}\alpha_i=0,\ 1\leq i\leq N.
\end{equation}
Denote one solution vector of these equations by $\bos$,
\begin{equation}
\bos=\sum_{\balpha\in R_+}\frac{\kappa(\balpha)}{\balpha\cdot\bos}\balpha.
\end{equation}
It is clear that $\bos\in\spn(R)$. Note that $s^2=\gamma$ because
\begin{equation}\label{normofs}
s^2=\bos\cdot\bos=\sum_{\balpha\in R_+}\frac{\kappa(\balpha)}{\balpha\cdot\bos}\bos\cdot\balpha=\sum_{\balpha\in R_+}\kappa(\balpha)=\gamma.
\end{equation}
The elements of the Hessian matrix $\bH(\bY)$ of $F_R(\bY)$ are given by
\begin{equation}\label{hessianmatrixcomponents}
[\bH(\bY)]_{ij}=\frac{\partial^2}{\partial Y_j\partial Y_i}F_R(\bY)=\delta_{ij}+\sum_{\balpha\in R_+}\frac{\kappa(\balpha)}{(\balpha\cdot\bY)^2}\alpha_i\alpha_j.
\end{equation}
$\bH(\bY)$ is a positive definite matrix for $\bY\cdot\balpha\neq 0$, because for $\bx\in\RR^N$,
\begin{equation}
\sum_{1\leq i,j\leq N}x_ix_j\frac{\partial^2}{\partial Y_j\partial Y_i}F_R(\bY)=x^2+\sum_{\balpha\in R_+}\frac{\kappa(\balpha)}{(\balpha\cdot\bY)^2}(\balpha\cdot\bx)^2\geq 0.
\end{equation}
Therefore, all the extrema of $F_R(\bY)$ are minima, and all eigenvalues of $\bH$ are larger than or equal to 1. Taking $\rho\in W$, one has
\begin{equation}
\rho \bos=\frac{1}{2}\sum_{\balpha\in R}\frac{\kappa(\balpha)}{\balpha\cdot\bos}\rho\balpha=\frac{1}{2}\sum_{\balpha^\prime\in R}\frac{\kappa(\balpha^\prime)}{\rho^{-1}\balpha^\prime\cdot\bos}\balpha^\prime=\frac{1}{2}\sum_{\balpha^\prime\in R}\frac{\kappa(\balpha^\prime)}{\balpha^\prime\cdot\rho\bos}\balpha^\prime.
\end{equation}
Here, the substitution $\balpha^\prime=\rho\balpha$ has been carried out. This means that $\rho\bos$ is also a solution of \eqref{ExtremaPotentialR}, and consequently, its solutions are related with each other by an element of the reflection group $W$. Therefore, there are $|W|$ solutions of \eqref{ExtremaPotentialR}, and they define the peak set of $R$. Because $F_R(\bY)$ is $W$-invariant, all the minima have the same value.  

Using the properties of the peak set, we construct an approximation of $\rme^{-\beta F_R(\bY)}/z_\beta$ when $\beta$ is very large using a second-order Taylor expansion. First, we choose an arbitrary element of the peak set, e.g.~$\bos$, and we approximate $z_\beta$ for large values of $\beta$ as follows.
\begin{equation}
z_\beta=\int_{\RR^N}\rme^{-\beta F_R(\bY)}\ud\bY\approx|W|\rme^{-\beta F_R(\bos)}\int_{\RR^N}\exp[-\beta \bor^T\bH(\bos)\bor/2]\ud\bor,
\end{equation}
where $\bor=\bY-\bos$. Because $\bH$ is positive definite and symmetric, and its eigenvalues are positive, we can use an orthogonal coordinate transformation to solve this Gaussian integral. The result is
\begin{equation}
z_\beta\approx |W|\rme^{-\beta F_R(\bos)} \prod_{i=1}^N\sqrt{\frac{2\pi}{\beta \lambda_i}},
\end{equation}
where the $\{\lambda_i\}_{i=1}^N$ are the eigenvalues of $\bH(\bos)$. Then, the following approximation holds,
\begin{equation}\label{SteadyStateDistributionLargeBetaApproximation}
\frac{\rme^{-\beta F_R(\bY)}}{z_\beta}\approx G_\beta(\bY),
\end{equation}
with $G_\beta(\bY)$ given by \eqref{gaussianapproximationFR}. Note that the approximate distribution is normalized. Finally, as $\beta\to\infty$, each of the Gaussians tends to a delta function in the sense of distributions,
\begin{equation}\label{steadydistributionbetalimit}
\lim_{\beta\to\infty}\frac{\rme^{-\beta F_R(\bY)}}{z_\beta}=\frac{1}{|W|}\sum_{\rho\in W}\delta^{(N)}(\bY-\rho\bos).
\end{equation} 
 
\bibliographystyle{aipnum4-1}
\bibliography{dsf_biblio}
 
\end{document}